\documentclass[11pt]{amsart}

\usepackage[T1]{fontenc}
\usepackage[utf8]{inputenc}
\usepackage{lmodern}
\usepackage[margin=1in]{geometry}
\usepackage{amsmath,amssymb,amsthm,mathtools}
\usepackage{thm-restate}
\usepackage{microtype}
\usepackage{float}
\usepackage{braket}
\usepackage{booktabs}
\usepackage[dvipsnames]{xcolor}
\usepackage[colorlinks, linkcolor=OliveGreen, citecolor=purple,urlcolor=purple, bookmarks=true]{hyperref}
\usepackage{algorithm} 
\usepackage[noend]{algpseudocode} % must be loaded before cleveref
\algnewcommand\algorithmicinput{\textbf{Input:}}
\algnewcommand\Input{\item[\algorithmicinput]}
\algnewcommand\algorithmicoutput{\textbf{Output:}}
\algnewcommand\Output{\item[\algorithmicoutput]}
\usepackage[capitalise,nameinlink]{cleveref}

\newcommand{\C}{\mathbb{C}}
\newcommand{\R}{\mathbb{R}}
\newcommand{\Trunc}{\mathsf{T}}
\newcommand{\init}{\mathrm{in}}
\newcommand{\vin}{v_{\init}}
\newcommand{\est}{\mathrm{est}}
\newcommand{\prep}{\mathrm{prep}}
\newcommand{\poly}{\mathrm{poly}}
\DeclareMathOperator{\supp}{supp}
\DeclarePairedDelimiter{\abs}{\lvert}{\rvert}
\DeclarePairedDelimiter{\norm}{\lVert}{\rVert}
\DeclarePairedDelimiterX{\ip}[2]{\langle}{\rangle}{#1,#2}

\theoremstyle{plain}
\newtheorem{theorem}{Theorem}[section]
\newtheorem{lemma}[theorem]{Lemma}
\theoremstyle{definition}
\newtheorem{definition}[theorem]{Definition}
\theoremstyle{remark}
\newtheorem{remark}[theorem]{Remark}

\floatstyle{ruled}
\restylefloat{algorithm}
\floatname{algorithm}{Algorithm}
\crefname{algorithm}{Algorithm}{Algorithms}
\Crefname{algorithm}{Algorithm}{Algorithms}

\numberwithin{equation}{section}
\hypersetup{
  pdftitle={Ground state computations by the truncated power method},
  pdfauthor={Savero Lukianto Chandra, Daochen Wang, Carolyn Zhang}
}

\makeatletter
\DeclareRobustCommand{\affiliationmark}[1]{%
  \hyperlink{tpm-affiliation-#1}{%
    \textsuperscript{\normalfont\@fnsymbol{#1}}}}
\newcommand{\affiliationnote}[2]{%
  \par\noindent\hangindent1.8em\hangafter1
  \makebox[1.8em][r]{%
    \hypertarget{tpm-affiliation-#1}{%
      \textsuperscript{\normalfont\@fnsymbol{#1}}}}%
  #2\par}
\renewcommand{\@setauthors}{%
  \begingroup
  \trivlist
  \centering\normalfont\footnotesize
  \@topsep16.5\p@\relax
  \advance\@topsep by -\baselineskip
  \item\relax
  \spaceskip=1.5\fontdimen2\font
  \MakeUppercase{\authors}
  \endtrivlist
  \endgroup}
\renewcommand{\@setthanks}{%
  \parindent\z@
  \def\thanks##1{##1}%
  \thankses}
\let\tpm@adminfootnotes\@adminfootnotes
\renewcommand{\@adminfootnotes}{%
  \long\def\@makefntext##1{\noindent##1}%
  \tpm@adminfootnotes}
\renewcommand{\footnoterule}{%
  \kern-3\p@\hrule\@width.5in\kern2.6\p@}
\makeatother
\title[Ground state computations by TPM]{%
  Ground state energy estimation and preparation\\
  by the truncated power method}

\author[S. L. Chandra, D. Wang, and C. Zhang]{%
  Savero Lukianto Chandra%
  \affiliationmark{1}\textsuperscript{,}%
  \affiliationmark{2}\textsuperscript{,}%
  \affiliationmark{3}\qquad
  Daochen Wang\affiliationmark{1}\qquad
  Carolyn Zhang%
  \affiliationmark{2}\textsuperscript{,}%
  \affiliationmark{3}}

\thanks{%
  \affiliationnote{1}{%
    Department of Computer Science,
    University of British Columbia, Vancouver, BC V6T 1Z4, Canada.}%
  \affiliationnote{2}{%
    Department of Physics and Astronomy,
    University of British Columbia, Vancouver, BC V6T 1Z1, Canada.}%
  \affiliationnote{3}{%
    Quantum Matter Institute,
    University of British Columbia,
    Vancouver, BC V6T 1Z4, Canada.}%
  \space
  Contact:
  S. L. Chandra
  (\href{mailto:saverolc@gmail.com}{\nolinkurl{saverolc@gmail.com}});
  D. Wang
  (\href{mailto:wdaochen@gmail.com}{\nolinkurl{wdaochen@gmail.com}});
  C. Zhang
  (\href{mailto:cczhang@phas.ubc.ca}{\nolinkurl{cczhang@phas.ubc.ca}}).%
}
\date{}

\begin{document}

\begin{abstract}
The Guided Local Hamiltonian Problem (GLHP) is the problem of ground state energy estimation and preparation given a guiding state having good overlap with the true ground state. The problem is known to be BQP-complete for small error but is classically tractable for larger error. In this work, we identify an orthogonal regime in which the problem is classically tractable, namely when the true ground state $\psi$ has polynomially bounded $\ell_1$-norm. In particular, we show that the Truncated Power Method (TPM), a classical algorithm introduced by Kirby et al.~\cite{accuracy_kirby_2026} in the quantum context, has a runtime that is polynomially comparable to quantum algorithms for GLHP in all relevant parameters, except for an extra factor of $\norm{\psi}_1^2$. In other words, if $\norm{\psi}_1^2$ is polynomially bounded, then TPM dequantizes quantum algorithms for GLHP.
\end{abstract}

\maketitle

%%%%%%%%%%%%%%%%%%%%%%
\section{Introduction}
\label{sec:intro}
%%%%%%%%%%%%%%%%%%%%%%

Ground state energy estimation and preparation are two of the most promising applications for quantum computation. Given an initial guiding state with good overlap with the true ground state, these tasks can be performed efficiently by quantum phase estimation (QPE) \cite{measurements_kitaev_1995}. This version of the problem, where a guiding state is available, has been the focus of many recent works from both the quantum algorithms and complexity communities \cite{preparation_lin_2020,heisenberg_lin_2022,dequantizing_gharibian_2023,hardness_cade_2023,guide_waite_2025,estimation_mande_2026,estimation_somma_2026,estimation_jeffery_2026,guiding_jeffery_2026,preparation_chen_2026}.
 
On the complexity side, Gharibian and Le Gall \cite{dequantizing_gharibian_2023} gave the first formal definition of the problem, which they called the Guided Local Hamiltonian Problem (GLHP). When the guiding state has constant overlap with the ground state, \cite{dequantizing_gharibian_2023} showed that the energy estimation version of GLHP is BQP-complete for small error but is classically tractable for larger error. In other words, when the error in the estimate is allowed to be sufficiently large, GLHP can be dequantized. Their result is striking because it delineates the boundary between the power of quantum and classical computing in terms of a single real parameter, namely the error.

In this work, we identify an orthogonal regime in which GLHP is classically tractable, namely when the true ground state $\psi$ has small $\ell_1$-norm in the computational basis, even if the desired error is small. More specifically, we show that a classical algorithm, the truncated power method (TPM), is polynomially comparable to quantum algorithms for GLHP in all relevant parameters (including the error) \emph{except} for an extra factor of $\norm{\psi}_1^2$. TPM was originally proposed by machine learning researchers Yuan and Zhang \cite{truncation_yuan_2013}, and introduced in the quantum context by Kirby et al.~\cite{accuracy_kirby_2026}.

For fair comparison between quantum and classical algorithms, we assume the input $n$-qubit Hamiltonian $H$ is given classically as a real linear combination of $M$ Pauli terms, i.e., $H=\sum_{j=1}^M \alpha_j P_j$. This representation is standard in condensed matter and quantum chemistry.  We also assume that the initial guiding state $v_{\init}$ is supported on $\poly(n)$ computational basis states. This allows us to simplify the comparison between quantum and classical costs and is without loss of generality in the sense that GLHP is BQP-hard even in this case \cite[Theorem 1.2]{dequantizing_gharibian_2023}.

In GLHP, we are given $v_{\init}$ with overlap $|\ip{v_{\init}}{\psi}|^2 \geq \delta$, where $\psi$ is the true ground state. We will assume the ground space is nondegenerate in the introduction for simplicity. In a standard estimation version of GLHP \cite{dequantizing_gharibian_2023}, the goal is to obtain an estimate of the ground state energy $E_0$ within $\varepsilon$ additive error. In a standard preparation version of GLHP \cite{preparation_lin_2020}, the goal is to prepare a state $v$ with $|\ip{v}{\psi}|^2 \geq 1 - \eta$, and we further assume that $E_1-E_0\geq \Delta$ where $E_0$ is the ground state energy and $E_1$ the first excited state energy (so $\Delta$ lower bounds the spectral gap). The parameter $\alpha \coloneqq \sum_{j=1}^M \abs{\alpha_j}$ upper bounds the norm of the Hamiltonian. The parameters $M$ and $\alpha$ should be thought of as scaling with $\poly(n)$.\footnote{More formally, we assume their polynomial scaling with $n$ compares to that of the support size $S$ of $v_{\init}$ in such a way that the $S$-dependent terms in \cref{tab:complexity-comparison} are subleading. This is only for convenience in presentation.} Under this setup, we compare in \cref{tab:complexity-comparison-intro} the performance of the classical TPM with the near-optimal quantum algorithm in \cite{preparation_lin_2020}. The results in the table do not technically need to assume $H$ is local; each Pauli term $P_j$ can act nontrivially on all $n$ qubits.

\begin{table}[htbp]
  \centering
  \setlength{\tabcolsep}{14pt}
  \setlength{\aboverulesep}{8pt}
  \setlength{\belowrulesep}{8pt}
  \begin{tabular}{lcc}
    \toprule
    Ground state task & Quantum algorithm (\cite{preparation_lin_2020}) & TPM (this work) \\
    \midrule
    Energy estimation
      & $\dfrac{M\alpha}{\sqrt{\delta}\,\varepsilon}$ 
      & $\norm{\psi}_1^2\cdot \dfrac{M\alpha^3}{\delta\,\varepsilon^3}$\\
    \addlinespace[8pt]
    Preparation
      & $\dfrac{M\alpha}{\sqrt{\delta}\,\Delta} \log \dfrac{1}{\eta}$ 
      & $\norm{\psi}_1^2 \cdot   \dfrac{M\alpha^3}{\Delta^3}
         \max\Bigl(\dfrac{1}{\delta}, \dfrac{1}{\eta}\Bigr)$\\
    \bottomrule
  \end{tabular}
  \vspace{8pt}
  \caption{Runtimes of TPM and the quantum algorithm of \cite{preparation_lin_2020} for estimating the ground state energy to additive error $\varepsilon$ and for preparing the ground state to fidelity at least $1-\eta$. Bounds hide constant and polylogarithmic factors. TPM is deterministic; the quantum algorithm succeeds with constant probability.}
  \label{tab:complexity-comparison-intro}
\end{table}

We make three remarks on the comparison in \cref{tab:complexity-comparison-intro}.
\begin{enumerate}
    \item We use Lin and Tong's algorithm \cite{preparation_lin_2020} as our main benchmark since their algorithm is near-optimal, by which we mean optimal up to logarithmic factors.\footnote{Recent works have removed even these logarithmic factors \cite{estimation_jeffery_2026,guiding_jeffery_2026,preparation_chen_2026}.} In particular, for the estimation task, Somma and de Wolf \cite{estimation_somma_2026} proved a lower bound of $\Omega(1/(\sqrt{\delta} \varepsilon))$ in terms of queries to $U\coloneqq \exp(iH)$ and $U^{-1}$, and a single use of $U$ or $U^{-1}$ costs $\widetilde O(M\alpha)$ gates using Hamiltonian simulation for unit time, see, e.g., \cite[circuit constructions and gate counts]{simulation_berry_2015}.
    \item Ref.~\cite{preparation_lin_2020} assumes a different access model to ours, namely a block-encoding of $H$. To make the comparison, we assume the block-encoding is obtained from the classical description of $H$ using the linear combination of unitaries technique \cite{lcu_childs_2012,simulation_berry_2015,algorithms_dalzell_2023}.
    \item For the preparation task, the complexities are technically incomparable. More specifically, \cite{preparation_lin_2020} outputs $v$ as a quantum state stored in qubits whereas TPM outputs $v$ as an explicit vector that can be written down on a piece of paper. Converting the quantum state to an explicit vector requires tomography. In \cref{sec:tomography},
    we show that this can be done by running the preparation algorithm $\widetilde{O}(\norm{\psi}_1^2/\eta^{3/2})$ times using results of \cite{tomography_vanapeldoorn_2022}, which we also show is optimal in the query model. Therefore, accounting for tomography makes the preparation runtimes between the quantum algorithm and TPM polynomially related.
\end{enumerate}

We now compare our results for the estimation task to those in \cite{accuracy_kirby_2026,dequantizing_gharibian_2023}. First, \cite[Theorem 1]{accuracy_kirby_2026} shows that the TPM algorithm has runtime
\begin{equation}\label{eq:l0-norm-bound}
    \widetilde O\Bigl(\norm{\psi}_0 \cdot 
\frac{M\alpha^3}{\Delta^3}\max\Bigl(\frac{1}{\delta},\frac{\alpha}{\varepsilon}\Bigr)\Bigr).
\end{equation}
Our energy estimation runtime in \cref{tab:complexity-comparison-intro} does not involve $\Delta$ dependence.\footnote{This may seem counterintuitive but the following argument suggests it is possible. If $\varepsilon\leq\Delta$, then the runtime implicitly depends on $\Delta$. On the other hand, if $\varepsilon > \Delta$, then the algorithm does not need to resolve the ground and first excited states to obtain an $\epsilon$-accurate ground state energy estimate, so it does not need to depend on $\Delta$.} However, we can get a $\Delta$-dependent runtime by preparing a state sufficiently close to the ground state (with sufficiently small $\eta$) to ensure an energy estimation error $\varepsilon$. Using the preparation runtime in \cref{tab:complexity-comparison-intro}, this gives an energy error of at most $2\norm{H}\eta$, which is upper bounded by $2\alpha \eta$.\footnote{More concretely, let $\ket{v} = \sqrt{1-p}\ket{\psi} + \sqrt{p}\ket{\phi}$ where $\ip{\psi}{\phi} = 0$ and $p\leq \eta$ to satisfy $|\ip{v}{\psi}|^2\geq 1-\eta$. Then evaluating the energy of $\ket{v}$ gives $\langle v|H|v\rangle - E_0=p(\langle \phi|H|\phi\rangle-E_0)$
Since both energies lie in $[-\|H\|,\|H\|]$, we get the loose upper bound $0\leq \langle v|H|v\rangle-E_0\leq 2\|H\|p\leq 2\|H\|\eta$.} Therefore, to obtain additive error $\varepsilon$, we can set $\eta$ as $\eta = \varepsilon/(2\alpha)$, which gives a runtime of
\begin{equation}\label{eq:l1-norm-bound}
    \widetilde O\Bigl(\norm{\psi}_1^2 \cdot 
\frac{M\alpha^3}{\Delta^3}\max\Bigl(\frac{1}{\delta},\frac{\alpha}{\varepsilon}\Bigr)\Bigr).
\end{equation}
Since $\norm{\psi}_1^2 \leq \norm{\psi}_0$, our bound in \cref{eq:l1-norm-bound} is always at least as good as \cref{eq:l0-norm-bound}, but it could be much better in cases when $\norm{\psi}_1^2$ is much smaller than $\norm{\psi}_0$. Indeed, it is trivial to turn a state $\psi$ with small $\norm{\psi}_1^2$ into a state with maximal $\norm{\psi}_0$ but still small $\norm{\psi}_1^2$ by adding a negligible but nonzero amount of amplitude on all basis states. We remark that \cref{eq:l1-norm-bound} is not comparable to the runtime stated for the estimation task in \cref{tab:complexity-comparison-intro}; the overall optimal runtime would be the smaller of (\ref{eq:l1-norm-bound}) or the $\Delta$-independent value in \cref{tab:complexity-comparison-intro}. In \cref{tab:complexity-comparison-intro} we use the $\Delta$-independent bound to clearly make the comparison to \cite{preparation_lin_2020}.

Second, \cite[Theorem 1.1]{dequantizing_gharibian_2023} yields a classical algorithm with runtime 
\begin{equation}
    \widetilde O (\alpha \varepsilon^{-1}\delta^{-2}(M2^k)^{\alpha \varepsilon^{-1}\log(1/\delta)}),
\end{equation}
where $k$ is the locality of the Hamiltonian. Therefore, $\varepsilon/\alpha$ has to be constant for the algorithm to be efficient, assuming constant $\delta$.\footnote{ Ref.~\cite{dequantizing_gharibian_2023} assumes a normalized $H$, which in our notation is $H/\alpha$. Their constant $\varepsilon$ translates to constant $\varepsilon/\alpha$ here.} In contrast, the complexities in   \cref{tab:complexity-comparison-intro} remain efficient for $\varepsilon/\alpha = 1/\poly(n)$.

\subsection{Technical overview.}
\label{sec:technical_overview}

TPM solves the ground state preparation and energy estimation problems by power iteration using the shifted Hamiltonian. For a Hamiltonian $H$ as above, the shifted Hamiltonian is $A = \alpha I - H$. Since $\alpha \geq \norm{H}$, the principal eigenvector of $A$ is the ground state $\psi$ of $H$. Therefore, applying $A$ to a vector amplifies its component in the direction of $\psi$.

TPM repeatedly applies $A$ to the initial state $v_{\init}$ to progressively approach $\psi$. Compared to textbook power iteration (see, e.g., \cite[Lecture 25]{algebra_trefethen_2022}), the \emph{only} change is that after each application of $A$, TPM truncates the vector to the $k$ basis states with the largest coefficients in absolute value. This truncation is crucial because without it, the number of parameters describing the evolving vector could grow exponentially, making the algorithm infeasible. After truncation, TPM normalizes the vector to preserve its unit norm, as also happens in textbook power iteration. This vector is then passed into the next iteration and the algorithm runs for a total of $L$ iterations.

The main challenge is showing that the iterates converge reasonably fast to $\psi$ when $k$ is of order $\norm{\psi}_1^2$. With this convergence rate in hand, the setting of $L$ naturally follows. That it suffices to set $k$ to be of order $\norm{\psi}_1^2$ is not a priori obvious. For both the estimation and preparation bounds, the key to obtaining a rigorous bound on $k$ is defining and tracking suitable progress measures.

Let us start with the preparation bound, as it is somewhat simpler conceptually (even though it is technically more demanding than the estimation bound). Following Kirby et al.~\cite{accuracy_kirby_2026}, we use the progress measure $T(v)$ to be the tangent of the angle between $\psi$ and a vector $v$. As the tangent of an angle, $T(v)$ does not depend on the normalization of $v$, which makes it convenient to work with. In particular, it does not change at the final normalization step of each iteration. When $T(v)$ is close to zero, $v$ is close to $\psi$. Therefore, it suffices to show that $T(v)$ shrinks considerably at each iteration.
When a vector $v$ is multiplied by $A$, we have $T(Av) \leq (1-\Delta/(2\alpha)) T(v)$ by a standard convergence analysis of power iteration \cite{accuracy_kirby_2026}. When $v$ is truncated to give $\Trunc_k(v)$, we show that
\begin{equation}\label{eq:truncation_effect}
    T(\Trunc_k(v)) \leq \frac{T(v) + \sqrt{2} \rho}{1-\rho T(v) - \rho^2},
\end{equation}
where $\rho \coloneqq \sqrt{\norm{\psi}_1^2/k}$. To parse this, note that when $k$ is very large, $\rho$ is approximately zero, and so the inequality becomes $T(\Trunc_k(v)) \leq T(v)$, which says that truncation did not hinder our progress towards $\psi$. In that case the analysis collapses to the usual analysis for power iteration. Of course, we cannot afford an arbitrarily large $k$ and the actual analysis requires more care. In particular, we show that the composed mapping $T(v) \mapsto T(\Trunc_k(Av))$ is contractive to a fixed point $\xi_\star$, that is of order $\rho \alpha/\Delta$, which is close to zero for sufficiently large $k$. The rate of contraction depends on the problem parameters, which in turn determines the setting of $L$. Working out the details gives the preparation bound in \cref{tab:complexity-comparison-intro}. Note that this analysis depends on $\Delta$ since the contraction rate depends on $\Delta$ even in the ideal case when no truncation occurs.

Since \cref{eq:truncation_effect} is crucial to why $k$ and our final bound scale with $\norm{\psi}_1^2$, let us sketch a proof here. Write $v = a \psi + b \psi^\perp$, where $\psi^\perp$ is a unit vector orthogonal to $\psi$. Write $\Pi$ for the projection onto the top $k$ largest coordinates of $v$. Then, we can control the effect of truncation on the inner product between $\psi$ and $v$:
\begin{equation}
    |\ip{\psi}{\Trunc_k(v)}| = |\ip{\psi}{\Pi v}| \geq |\ip{\psi}{v}| - |\ip{\psi}{(I-\Pi)v}| \geq \abs{a} - \norm{\psi}_1\norm{(I-\Pi)v}_{\infty},
\end{equation}
where $\tau \coloneqq \norm{(I-\Pi)v}_{\infty}$ is the largest absolute value of an entry of $(I-\Pi)v$. But by the definition of $\Pi$, every nonzero entry of $\Pi v$ has magnitude at least $\tau$, and so
\begin{equation}
    \tau k \leq \norm{\Pi v}_1 \leq \abs{a} \norm{\Pi\psi}_1 + \abs{b} \norm{\Pi \psi^\perp}_1 \leq \abs{a} \norm{\psi}_1 + \abs{b} \sqrt{k},
\end{equation}
where the last inequality is because $\Pi\psi^\perp$ has at most $k$ nonzero entries. Combining the previous two equations gives
\begin{equation}\label{eq:cosine}
     |\ip{\psi}{\Trunc_k(v)}| \geq \abs{a} - \abs{a}\rho^2 - \abs{b}\rho = \abs{a}(1-\rho T(v) - \rho^2).
\end{equation}
Some elementary algebra then converts this cosine-type bound to the tangent bound in \cref{eq:truncation_effect}.

For the energy estimation bound, we need a different analysis to bypass the dependency on $\Delta$. We now track two progress measures: 1. $R(v) \coloneqq \ip{v}{Av}/\ip{v}{v}$, the Rayleigh quotient of a vector $v$; and 2. $C(v)$, the cosine of the angle between $\psi$ and $v$.\footnote{Of course, cosine and tangent are interconvertible. The choice of cosine here makes the algebra much nicer.} Note that $R(v)$ is at most $\norm{A}$ and the closer $R(v)$ is to $\norm{A}$, the better the energy estimate corresponding to $v$ is. When $v$ is multiplied by $A$, we show that
\begin{equation}
    C(Av) \geq \sqrt{\frac{\norm{A}}{R(v)}} C(v).
\end{equation}
The intuition here is that if $R(v)$ is small, then $v$ has more amplitude on eigenvectors of $A$ with smaller eigenvalues. Therefore, when $v$ is multiplied by $A$, the relative amplification of $v$'s component in the direction of $\psi$ is larger. This explains why the amplification factor, i.e., $\sqrt{\norm{A}/R(v)}$, is larger for smaller $R(v)$. Following multiplication by $A$, TPM performs truncation. The effect of truncation can be controlled using a bound like \cref{eq:cosine} with $k$ sufficiently large ($\rho$ sufficiently small). Since we have already covered this, let us continue as though no truncation occurs. 

Now, as long as $R(v)$ is $\varepsilon$-\emph{far} from $\norm{A}$, then the amplification factor is at least $\sqrt{\norm{A}/(\norm{A} - \varepsilon)} = (1-\varepsilon/\norm{A})^{-1/2} \geq (1 - \varepsilon/(2\alpha))^{-1/2}$, where the last inequality uses $\norm{A} \leq 2\alpha$. Therefore, after a certain number of iterations $L$, the value of $C(v)$ would exceed $1$. But that is impossible as $C(v)$ represents a cosine. Therefore, it must be the case that prior to iteration $L$, the value of $R(v)$ is $\varepsilon$-\emph{close} to $\norm{A}$. Working out the details gives the estimation bound in \cref{tab:complexity-comparison-intro}. Note that this analysis used $\varepsilon$ to control the progress and is independent of $\Delta$.

%%%%%%%%%%%%%%%%%%%%%%%%%%%%%%%%%%%%
\section{Algorithm and main results}
\label{sec:algorithm}
%%%%%%%%%%%%%%%%%%%%%%%%%%%%%%%%%%%%

\subsection{Notation and setup}
\label{sec:notation}

Norms and supports refer to the computational basis; logarithms are
natural. For $x,y\in\C^N$, write $\ip{x}{y}=x^\dagger y$,
$\norm{x}_1=\sum_{j=1}^N\abs{x_j}$, and
$\norm{x}_2=(\sum_{j=1}^N\abs{x_j}^2)^{1/2}$. We use the operator norm $\norm{C}$ and write $C\preceq D$
when $D-C$ is positive semidefinite. We sometimes put vector symbols into bras and kets to stress that the vector should be interpreted as a quantum state.

Let $H\in\C^{N\times N}$ be an $n$-qubit Hamiltonian, where $N=2^n$, with
Pauli decomposition
\begin{equation}\label{eq:pauli}
  H=\sum_{j=1}^M h_jP_j,\qquad h_j\in\R.
\end{equation}
Each $P_j$ is a tensor product of matrices in $\{I,X,Y,Z\}$. Let $E_0$ be the ground state energy, or equivalently, the smallest eigenvalue of $H$, and $\Omega \coloneqq\ker(H-E_0I)$ be the ground space of $H$. We define $P_\Omega$ as the orthogonal projection onto $\Omega$. Since TPM involves vector truncation, we define the sparsity parameter
\begin{equation}\label{eq:ground-space-parameters}
  \chi_\Omega\coloneqq
  \max_{u\in\Omega\colon \norm{u}_2=1} \, \norm{u}_1^2,
\end{equation}
that captures the compressibility of the states in the ground space, which will be useful to quantify the effect of truncation. The maximum exists because we consider unit vectors in the finite-dimensional space $\Omega$. When $H$ has a unique ground state, $\Omega$ is spanned by a normalized vector $\psi$ so that $\chi_\Omega=\norm{\psi}_1^2$.

In power iteration, applying a matrix to a vector gives the greatest relative amplification to the component along the eigenvector associated with the eigenvalue of largest magnitude. To establish a maximum amplification of the eigenvector associated with $E_0$, we shift the Hamiltonian so that the shifted eigenvalue corresponding to $E_0$ is guaranteed to have the highest magnitude. We define the following variables.
\begin{equation}\label{eq:shift}
  \alpha\coloneqq\sum_{j=1}^M\abs{h_j},\qquad
  A\coloneqq\alpha I-H,\qquad
  \lambda\coloneqq\alpha-E_0.
\end{equation}
Since $\norm{P_j}=1$, we have $\norm{H}\leq\alpha$. Consequently,
\begin{equation}\label{eq:spectral-bounds}
  0\preceq A\preceq 2\alpha I,\qquad
  AP_\Omega=\lambda P_\Omega,\qquad
  \lambda=\norm{A}\leq 2\alpha.
\end{equation}
Thus, an estimate $\mu$ of $\lambda$ gives the ground state energy estimate $\widehat E_0=\alpha-\mu$.

Another parameter involved in the theoretical analysis is the spectral gap of the shifted Hamiltonian, defined as the ratio between the difference of the two largest distinct eigenvalues and the largest eigenvalue:
\begin{equation}\label{eq:preparation-gap}
  \gamma\coloneqq\frac{E_+-E_0}{\lambda} \leq 1,
\end{equation}
where $E_+$ denotes the smallest eigenvalue of $H$ that is strictly larger than $E_0$.

All bounds in the main theorems (\cref{thm:estimation} and \cref{thm:preparation}) below use the coefficient sum $\alpha$ in \cref{eq:shift}, which is available directly from the input. We remark that a larger shift brings the eigenvalue ratios of $A$ closer to one and can slow power iteration. Thus, in practice, one may prefer a tighter upper bound on the largest eigenvalue of $H$.

\subsection{Algorithm}
\label{sec:subsec_algorithm}

In this subsection, we present the Truncated Power Method (TPM) algorithm for finding the minimum energy of $H$ by estimating the principal eigenvalue of the shifted Hamiltonian $A$ as in \cref{eq:shift}. The inputs of the algorithm are the shifted Hamiltonian $A$, an initial guess $v_{\init}$, the truncation threshold $k$, and the number of power iterations $L$. In order to make progress, it is standard to assume that the initial ground space overlap is positive:
\begin{equation}\label{eq:initial-overlap}
  \delta_\Omega\coloneqq\norm{P_\Omega\vin}_2^2>0.
\end{equation}

The TPM pseudocode is shown in \Cref{alg:tpm}. For a positive integer $k$, $\Trunc_k(x)$ is defined to be the vector truncation operator that retains the $k$ basis states of largest absolute value in $x\in\C^N$. For simplicity, we suppose the truncation resolves ties by a fixed coordinate ordering. The algorithm runs for $L$ iterations, and outputs a real number $\mu$ and a vector $v_L$ which respectively approximate the actual principal eigenvalue of $A$ and the eigenvector associated with it.

\begin{algorithm}[htbp]
\caption{Truncated power method (TPM)}\label{alg:tpm}
\begin{algorithmic}[1]
\Input $A\in\C^{N\times N}$ with $A\succeq0$;
$\vin\in\C^N$ with $\norm{\vin}_2=1$; positive integers $k,L$.
\Output A pair $(\mu,v_L)$ with $\norm{v_L}_2=1$ and
$\mu=\ip{v_L}{Av_L}$ estimating $\norm{A}$.
\State $v_0\gets\vin$
\For{$t=0,1,\ldots,L-1$}
    \State $y_t\gets Av_t$
    \State $z_t\gets\Trunc_k(y_t)$
    \If{$z_t=0$}
        \State $v_L\gets v_t$
        \State \Return $(0,v_L)$
    \EndIf
    \State $v_{t+1}\gets z_t/\norm{z_t}_2$
\EndFor
\State $\mu\gets\ip{v_L}{Av_L}$
\State \Return $(\mu,v_L)$
\end{algorithmic}
\end{algorithm}

\subsection{Main results}
\label{sec:results} 

This subsection presents two theorems on the correctness of \cref{alg:tpm}. \cref{thm:estimation} concerns the eigenvalue estimation, which specifies the truncation threshold $k_\est$ and minimum number of iterations $L_\est$ sufficient to approximate the principal eigenvalue to absolute error $\varepsilon$. Meanwhile, \cref{thm:preparation} covers the ground state preparation, specifying the critical resources for estimating the true ground state to fidelity $1-\eta$ for some $\eta>0$.

\begin{theorem}[Ground state energy estimation]\label{thm:estimation}
Suppose $\delta_\Omega>0$. For an accuracy target $\varepsilon >0$, define
\begin{equation}\label{eq:parameters}
  k_\est\coloneqq
  \frac{64\chi_\Omega\alpha^2}{\delta_\Omega\varepsilon^2},
  \qquad
  L_\est\coloneqq
  1+\frac{4\alpha}{\varepsilon}\log\frac{2}{\sqrt{\delta_\Omega}}.
\end{equation}
For integers $k\geq k_\est$ and $L\geq L_\est$, the output
$(\mu,v_L)$ of \cref{alg:tpm} satisfies
\begin{equation}\label{eq:energy-guarantee}
  \alpha-E_0-\varepsilon\leq\mu\leq\alpha-E_0.
\end{equation}
Equivalently, $\widehat E_0\coloneqq\alpha-\mu$ satisfies
$E_0\leq\widehat E_0\leq E_0+\varepsilon$.
\end{theorem}

\begin{theorem}[Ground state preparation]\label{thm:preparation}
Suppose $\delta_\Omega>0$. For an infidelity target  $0<\eta<1$, define
\begin{equation}\label{eq:preparation-parameters}
  k_{\prep}\coloneqq
  \frac{\chi_\Omega}{\gamma^2}
  \max\Bigl\{\frac{16}{\delta_\Omega},\frac{25}{\eta}\Bigr\},
  \qquad
  L_{\prep}\coloneqq
  \frac{4}{\gamma}\log\frac{25}{\delta_\Omega\eta}.
\end{equation}
For integers $k\geq k_{\prep}$ and $L\geq L_{\prep}$,
the output $(\mu,v_L)$ of \cref{alg:tpm} satisfies
\begin{equation}\label{eq:preparation-guarantee}
  \norm{P_\Omega v_L}_2^2\geq1-\eta.
\end{equation}
Moreover, $\lambda(1-\eta)\leq\mu\leq\lambda$.
\end{theorem}

The proofs of \cref{thm:estimation} and \cref{thm:preparation} appear in \cref{sec:proof} and \cref{sec:preparation} respectively. We deliberately design the \cref{thm:estimation} proof to be spectral-gap-independent and \cref{thm:preparation} to be the opposite to comply with the dequantization task in \cref{sec:dequantization}.

\subsection{Sparse arithmetic cost}
\label{sec:cost}

We present a time complexity analysis of \cref{alg:tpm} and show that it is polynomial in the input parameters. As we have seen in \cref{sec:notation}, we assume that $A$ is given in the Pauli decomposition form with $M+1$ terms, where the extra $+1$ comes from the shift, and $\vin$ is supplied as a list of its $S\coloneqq\abs{\supp(\vin)}$ nonzero amplitudes and basis labels. A key point here is that each Pauli string simply permutes computational basis states up to phases, so we count the application of a Pauli string to a computational basis state as a unit-cost operation. A na\"ive representation of Pauli strings as
exponentially-sized matrices would induce exponentially expensive operation cost for the matrix-vector product. We also count amplitude arithmetic and comparisons between two basis vectors as unit-cost operations.

Applying $A$ to a vector supported on $s$ computational basis states generates a vector supported on at most $(M+1)s$ basis states. The entire process of adding up the coefficients for the same basis states, sorting based on magnitude for truncation, discarding low-weight basis states, and normalizing the result costs $\widetilde O(Ms)$ operations. We estimate that the first iteration costs $\widetilde O(MS)$; each subsequent
iteration and the final Rayleigh quotient cost
$\widetilde O(Mk)$ because the state is capped to be $k$-sparse. Since there are a total of $L$ power iterations, \cref{alg:tpm} returns $(\mu,v_L)$ in
\begin{equation}\label{eq:general-runtime}
  \widetilde O\bigl(MS+Mk L\bigr)
\end{equation}
sparse arithmetic operations.

%%%%%%%%%%%%%%%%%%%%%%%%%%%%%%%%%%%%%%%%
\section{Ground state energy estimation}
\label{sec:proof}
%%%%%%%%%%%%%%%%%%%%%%%%%%%%%%%%%%%%%%%%

In this section, we prove \cref{thm:estimation} by tracking the following progress measures. The quantity $C_\Omega(v)$ can be thought of as the ``cosine of the angle between $v$ and $\Omega$,'' while $R(v)$ is the Rayleigh quotient of $A$ with respect to $v$. 

\begin{definition}[Progress measures]\label{def:estimation-progress}
For a nonzero vector $v\in\C^N$, define
\begin{equation}\label{eq:estimation-progress}
  C_\Omega(v)\coloneqq\frac{\norm{P_\Omega v}_2}{\norm{v}_2},\qquad
  R(v)\coloneqq\frac{\ip{v}{Av}}{\ip{v}{v}}.
\end{equation}
\end{definition}
Both quantities are invariant under multiplication by a nonzero scalar. Thus, normalization after truncation does not affect either quantity and makes that step easy to analyze. We first establish bounds for the power iteration and truncation steps, and prove the monotonicity of the Rayleigh quotient in \cref{sec:aux-est}. Then, we combine these lemmas to prove \cref{thm:estimation} in \cref{sec:proof-est}.

Throughout this section, let $k$ be a positive integer and set
\begin{equation}\label{eq:estimation-rho}
  \rho\coloneqq\sqrt{\frac{\chi_\Omega}{k}},
\end{equation}
where $\chi_\Omega$ is defined in \cref{eq:ground-space-parameters}.

\subsection{Auxiliary estimates}
\label{sec:aux-est}

The first two lemmas play the same roles as \cref{lem:subspace-power,lem:subspace-truncation}: the power step increases the overlap, while truncation reduces it by at most $\rho$. The amplification in the power step depends on the current Rayleigh quotient rather than the spectral gap.

\begin{lemma}[Power iteration]\label{lem:estimation-power}
For every nonzero vector $v$ with $R(v)>0$, we have
\begin{equation}\label{eq:estimation-power}
  C_\Omega(Av)\geq\sqrt{\frac{\lambda}{R(v)}}\,C_\Omega(v).
\end{equation}
\end{lemma}

\begin{proof}
Since $A$ commutes with $P_\Omega$ and acts as multiplication by $\lambda$ on $\Omega$,
\begin{equation}
  \norm{P_\Omega Av}_2=\lambda\norm{P_\Omega v}_2.
\end{equation}
Also, $0\preceq A\preceq\lambda I$ implies $A^2\preceq\lambda A$, so
\begin{equation}
  \norm{Av}_2^2=\ip{v}{A^2v}
  \leq\lambda\ip{v}{Av}
  =\lambda R(v)\norm{v}_2^2.
\end{equation}
The assumption $R(v)>0$ ensures that $Av\ne0$. Dividing the first identity by the square root of the second bound gives \cref{eq:estimation-power}.
\end{proof}

\begin{lemma}[Truncation]\label{lem:estimation-truncation}
For every vector $y$ with $\Trunc_k(y)\ne0$, we have
\begin{equation}\label{eq:estimation-truncation}
  C_\Omega(\Trunc_k(y))\geq C_\Omega(y)-\rho.
\end{equation}
\end{lemma}

\begin{proof}
Write $z\coloneqq\Trunc_k(y)$ and $d\coloneqq y-z$. If $d=0$, the result is immediate. Otherwise, every discarded coordinate has magnitude at most the smallest retained magnitude, which is at most $\norm{z}_2/\sqrt{k}$. By the definition of $\chi_\Omega$ and H\"older's inequality,
\begin{equation}\label{eq:projected-tail}
  \norm{P_\Omega d}_2
  =\max_{u\in\Omega\colon\norm{u}_2=1} \, \abs{\ip{u}{d}}
  \leq\sqrt{\chi_\Omega}\norm{d}_\infty
  \leq\rho\norm{z}_2.
\end{equation}
The triangle inequality therefore gives
\begin{equation}
  \norm{P_\Omega z}_2
  \geq\norm{P_\Omega y}_2-\rho\norm{z}_2.
\end{equation}
Dividing by $\norm{z}_2$ and using $\norm{z}_2\leq\norm{y}_2$, we obtain
\begin{equation}
  C_\Omega(z)\geq\frac{\norm{P_\Omega y}_2}{\norm{z}_2}-\rho
  \geq C_\Omega(y)-\rho.\qedhere
\end{equation}
\end{proof}

The last lemma shows that a full iteration does not decrease the Rayleigh quotient once the input is $k$-sparse, as it is from the $t=1$ iteration onwards. Intuitively, after normalizing $v$, the normalized truncation of $Av$
is the $k$-sparse unit vector most aligned with $Av$. Since $v$
itself is an eligible choice, the new vector is at least as aligned
with $Av$ as $v$ is. Positive semidefiniteness of $A$ then converts this alignment comparison into a nondecreasing Rayleigh quotient.

\begin{lemma}[Rayleigh-quotient monotonicity]\label{lem:monotonicity}
For every nonzero $k$-sparse vector $v$ with $R(v)>0$, the vector $\Trunc_k(Av)$ is nonzero and
\begin{equation}\label{eq:monotonicity}
  R(\Trunc_k(Av))\geq R(v).
\end{equation}
\end{lemma}

\begin{proof}
By scale invariance, we may assume that $\norm{v}_2=1$.
Since $R(v)>0$, we have $Av\ne0$, and hence
$z\coloneqq\Trunc_k(Av)\ne0$. Set
$w\coloneqq z/\norm{z}_2$.

Let $\Pi$ be the coordinate projection onto $\supp(v)$.
Since $v$ is $k$-sparse and truncation retains the largest
coordinates of $Av$, we have
\begin{equation}
  \ip{v}{Av}
  =\ip{v}{\Pi Av}
  \leq\norm{\Pi Av}_2
  \leq\norm{z}_2
  =\ip{w}{Av},
\end{equation}
where the last equality follows from the orthogonality of
the retained and discarded coordinates. 

Since $v$ and $w$ are normalized, we have
\begin{equation}
  R(w)-R(v)
  =\ip{w}{Aw}-\ip{v}{Av} =2\ip{w-v}{Av}
    +\ip{w-v}{A(w-v)}
  \geq0,
\end{equation}
where the last term is nonnegative because $A\succeq0$.
Finally, $R(w)=R(\Trunc_k(Av))$ by scale invariance,
which proves the lemma.
\end{proof}

\subsection{TPM correctness proof for ground state energy estimation}
\label{sec:proof-est}

\begin{proof}[Proof of \cref{thm:estimation}]
Every output of \cref{alg:tpm} satisfies $0\leq\mu\leq\lambda$, including the early exit with $\mu=0$. Thus, the result is immediate if $\varepsilon\geq\lambda$. 

We therefore focus on the case $\varepsilon < \lambda$, and define the target relative error as
\begin{equation}\label{eq:normalized-parameters}
  \zeta\coloneqq\frac{\varepsilon}{\lambda}\in(0,1).
\end{equation}
For each iterate, write
\begin{equation}
  c_t\coloneqq C_\Omega(v_t),\qquad R_t\coloneqq R(v_t),\qquad
  c_0=\sqrt{\delta_\Omega}.
\end{equation}
Since $\lambda\leq2\alpha$, the choices of $k$ and $L$ in \cref{eq:parameters} imply
\begin{equation}\label{eq:rho-bound}
  k\geq\frac{16\chi_\Omega}{\delta_\Omega\zeta^2},
  \qquad
  \rho\leq\frac{\zeta c_0}{4},
\end{equation}
and
\begin{equation}\label{eq:L-bound}
  L-1\geq\frac{4\alpha}{\varepsilon}\log\frac{2}{c_0}
  \geq\frac{2}{\zeta}\log\frac{2}{c_0}.
\end{equation}

We first check that all $L$ iterations are well defined. Since $A\succeq\lambda P_\Omega$, we have $R_0\geq\lambda c_0^2>0$, so $Av_0$ and $\Trunc_k(Av_0)$ are nonzero. By \cref{lem:estimation-power,lem:estimation-truncation} and $R_0\leq\lambda$,
\begin{equation}\label{eq:estimation-first-step}
  c_1\geq\sqrt{\frac{\lambda}{R_0}}\,c_0-\rho
  \geq c_0-\rho>0.
\end{equation}
Consequently, $R_1\geq\lambda c_1^2>0$. Every subsequent iterate is $k$-sparse, so induction using \cref{lem:monotonicity} rules out the zero-output branch and gives
\begin{equation}\label{eq:ordered-rayleigh}
  0<R_1\leq R_2\leq\cdots\leq R_L\leq\lambda.
\end{equation}
The initial vector need not be $k$-sparse, which is why we treat the first iteration separately.

It remains to show that $R_L\geq\lambda-\varepsilon$. Suppose, for a contradiction, that $R_L<\lambda-\varepsilon$. By \cref{eq:ordered-rayleigh}, $R_t<\lambda-\varepsilon$ for $1\leq t\leq L$. Define
\begin{equation}\label{eq:estimation-growth-factor}
  g\coloneqq\sqrt{\frac{\lambda}{\lambda-\varepsilon}}
  =(1-\zeta)^{-1/2}>1.
\end{equation}
Combining the power and truncation bounds (\cref{lem:estimation-power,lem:estimation-truncation}) now gives
\begin{equation}\label{eq:overlap-recurrence}
  c_{t+1}\geq\sqrt{\frac{\lambda}{R_t}}\,c_t-\rho
  \geq g c_t-\rho,\qquad 1\leq t<L.
\end{equation}
The affine map $c\mapsto gc-\rho$ has a repelling fixed point
\begin{equation}\label{eq:estimation-fixed-point}
  c_\star\coloneqq\frac{\rho}{g-1}>0,
  \qquad
  c_{t+1}-c_\star\geq g(c_t-c_\star).
\end{equation}
To bound the initial distance from this fixed point, observe that
\begin{equation}
  \rho+c_\star=\frac{\rho}{1-\sqrt{1-\zeta}}
  \leq\frac{2\rho}{\zeta}\leq\frac{c_0}{2}.
\end{equation}
Together with \cref{eq:estimation-first-step}, this implies
\begin{equation}
  c_1-c_\star\geq c_0-(\rho+c_\star)\geq\frac{c_0}{2}.
\end{equation}
Iterating \cref{eq:estimation-fixed-point} and using
$\log g=-\tfrac12\log(1-\zeta)\geq\zeta/2$, we obtain
\begin{equation}\label{eq:unrolled}
  c_L\geq c_\star+g^{L-1}(c_1-c_\star)
  \geq c_\star+e^{\zeta(L-1)/2}\frac{c_0}{2}
  \geq c_\star+1>1,
\end{equation}
where the last bound follows from \cref{eq:L-bound}. This contradicts $c_L\leq1$. Therefore, $R_L\geq\lambda-\varepsilon$, and $\mu=R_L\leq\lambda$ completes the proof.
\end{proof}

\begin{remark}[A single reference ground state]\label{rem:reference-state}
For a normalized $\psi\in\Omega$ with
$\delta_\psi\coloneqq\abs{\ip{\psi}{\vin}}^2>0$, the energy guarantee
holds with $\chi_\Omega$ replaced by $\norm{\psi}_1^2$ and
$\delta_\Omega$ by $\delta_\psi$. The proof uses only the projection
property and $AP_\Omega=\lambda P_\Omega$, so it applies to
$\psi\psi^\dagger$ in place of $P_\Omega$. Choosing
$\psi=P_\Omega\vin/\sqrt{\delta_\Omega}$ preserves the initial overlap
and may reduce the sparsity parameter. The formulation with $P_\Omega$ uses the common parameters of
\cref{thm:estimation,thm:preparation}.
\end{remark}

%%%%%%%%%%%%%%%%%%%%%%%%%%%%%%%%%%
\section{Ground state preparation}
\label{sec:preparation}
%%%%%%%%%%%%%%%%%%%%%%%%%%%%%%%%%%

In this section, we prove \cref{thm:preparation} by tracking the following progress measure. The quantity $T_\Omega(v)$ can be thought of as the ``tangent of the angle between $v$ and $\Omega$''. This measure was originally defined and used by Kirby et al. in \cite{accuracy_kirby_2026}.
\begin{definition}[Progress measure \cite{accuracy_kirby_2026}]\label{def:subspace-progress}
For a vector $v\in\C^N$ with $P_\Omega v\ne0$, define
\begin{equation}\label{eq:subspace-progress}
  T_\Omega(v)\coloneqq
  \frac{\norm{(I-P_\Omega)v}_2}{\norm{P_\Omega v}_2}.
\end{equation}
\end{definition}
Like $C_{\Omega}(v)$ in the previous section, $ T_\Omega(v)$ is invariant under multiplication of $v$ by a nonzero scalar. In \cref{sec:aux-prep}, we establish some lemmas that bound the tangent at each step in \cref{alg:tpm}. We then use the lemmas to prove \cref{thm:preparation} in \cref{sec:proof-prep}.

Throughout this section, let $k$ be a positive integer and set
\begin{equation}\label{eq:preparation-rho}
  \rho\coloneqq\sqrt{\frac{\chi_\Omega}{k}},
\end{equation}
where $\chi_\Omega$ is defined in \cref{eq:ground-space-parameters}. This is as in the previous section.

\subsection{Auxiliary estimates}
\label{sec:aux-prep}

We introduce three helper lemmas for proving \cref{thm:preparation}. Broadly speaking, \cref{lem:subspace-power} describes by how much the tangent factor contracts, or equivalently, how much closer we get to the ground space, after each application of $A$. Following this, \cref{lem:subspace-truncation} characterizes how much ground space overlap we give away due to truncation, which translates into an upper bound on the progress measure associated with the truncated state. Finally, \cref{lem:preparation-recurrence} shows that the progress measure at each iteration is upper bounded by a sequence that contracts to a fixed point near zero for sufficiently large $k$. This means that \cref{alg:tpm} steadily approaches the ground space for sufficient large $k$.

\cref{lem:subspace-power} already appears as \cite[Lemma 2]{accuracy_kirby_2026} but we give a proof simply for completeness. \cref{lem:subspace-truncation} is new to this work as it works with the $\ell_1$ norm of the ground state, whereas its corresponding lemma in \cite{accuracy_kirby_2026}, namely Lemma 3, works with the $\ell_0$ norm. \cref{lem:preparation-recurrence} can be viewed as modularizing parts of \cite[Proof of Theorem 1]{accuracy_kirby_2026} into a self-contained lemma about how a specific function contracts.

We prove \cref{lem:subspace-power,lem:subspace-truncation} here and \cref{lem:preparation-recurrence} in \cref{app:preparation-contraction}.

\begin{lemma}[{Power iteration \cite[Lemma 2]{accuracy_kirby_2026}}]\label{lem:subspace-power}
For every nonzero vector $v$, if $P_\Omega v\ne0$, then
\begin{equation}\label{eq:subspace-power}
  T_\Omega(Av)\leq(1-\gamma)T_\Omega(v).
\end{equation}
\end{lemma}

\begin{proof} We give a proof for completeness. We bound the components of $Av$ in $\Omega$ and $\Omega^\perp$ separately. Since $A$ acts as multiplication by $\lambda$ on $\Omega$ and commutes with $P_\Omega$, we have
\begin{equation}
  P_\Omega Av=AP_\Omega v=\lambda P_\Omega v.
\end{equation}
Taking norms on both sides, the ground-space component satisfies
\begin{equation}
  \norm{P_\Omega Av}_2=\lambda\norm{P_\Omega v}_2.
\end{equation}

Meanwhile, the eigenvalues of $A$ on $\Omega^\perp$ lie between
zero and $\alpha-E_+=\lambda-(E_+-E_0)=\lambda(1-\gamma)$. Therefore, the component orthogonal to the ground space satisfies
\begin{equation}
  \norm{(I-P_\Omega)Av}_2
  =\norm{A(I-P_\Omega)v}_2
  \leq\lambda(1-\gamma)\norm{(I-P_\Omega)v}_2.
\end{equation}

Finally, by the definition of $T_\Omega$, we obtain
\begin{equation}
  T_\Omega(Av)
  =\frac{\norm{(I-P_\Omega)Av}_2}{\norm{P_\Omega Av}_2}
  \leq\frac{\lambda(1-\gamma)\norm{(I-P_\Omega)v}_2}
              {\lambda\norm{P_\Omega v}_2}
  =(1-\gamma)T_\Omega(v).
\end{equation}
\end{proof}

The next lemma bounds the effect of truncation on ground space overlap and progress.

\begin{lemma}[Truncation]\label{lem:subspace-truncation}
Let $\theta$ be a vector with $P_\Omega\theta\ne0$, and set
\begin{equation}
  a\coloneqq\norm{P_\Omega\theta}_2,\qquad
  x\coloneqq T_\Omega(\theta),\qquad
  \omega\coloneqq\Trunc_k(\theta).
\end{equation}
Then
\begin{equation}\label{eq:subspace-overlap}
  \norm{P_\Omega\omega}_2\geq a(1-\rho x-\rho^2).
\end{equation}
Moreover, if $\rho x+\rho^2<1$, then $P_\Omega\omega\neq0$ and
\begin{equation}\label{eq:subspace-truncation}
  T_\Omega(\omega)
  \leq\frac{x+\sqrt2\rho}{1-\rho x-\rho^2}.
\end{equation}
\end{lemma}

\begin{proof}
For convenience, we may assume $\theta$ is normalized without loss of generality. We first prove \cref{eq:subspace-overlap} and then \cref{eq:subspace-truncation}.

We begin by writing
\begin{equation}\label{eq:preparation-decomposition}
  \theta=au+bu^\perp,\qquad
  u\coloneqq\frac{P_\Omega\theta}{a},\qquad
  b\coloneqq\norm{(I-P_\Omega)\theta}_2,
\end{equation}
where $u\in\Omega$ and $u^\perp\in\Omega^\perp$ are normalized and $x=b/a$. By the definition of $\chi_\Omega$, we have $\norm{u}_1\leq\sqrt{\chi_\Omega}$.

Let $F$ be the set of the $k$ retained coordinates, and let $\Pi_F$ be the coordinate projection onto $F$. Since $\omega=\Pi_F\theta$ and $\ip{u}{\theta}=a$, the triangle inequality gives
\begin{equation}
  \norm{P_\Omega\omega}_2
  \geq\abs{\ip{u}{\omega}}
  \geq \abs{\ip u\theta} - \abs{\ip {u}{(I-\Pi_F)\theta}} =a-\abs{\ip{u}{(I-\Pi_F)\theta}}.
\end{equation}
It therefore suffices to prove
\begin{equation}\label{eq:preparation-overlap-target}
  \abs{\ip{u}{(I-\Pi_F)\theta}}
  \leq a(\rho x+\rho^2).
\end{equation}
The $(1,\infty)$ case of H\"older's inequality gives
\begin{equation}\label{eq:preparation-holder}
  \abs{\ip{u}{(I-\Pi_F)\theta}}
  \leq\norm{u}_1\norm{(I-\Pi_F)\theta}_\infty.
\end{equation}
Define the truncated basis state magnitude threshold
\begin{equation}
  \tau\coloneqq\min_{j\in F}\abs{\theta_j}.
\end{equation}
Every discarded coordinate has absolute value at most $\tau$, so
$\norm{(I-\Pi_F)\theta}_\infty\leq\tau$. To bound $\tau$, use
\cref{eq:preparation-decomposition} and the fact that $F$ has
cardinality $k$ to obtain
\begin{equation}
  k\tau\leq\sum_{j\in F}\abs{\theta_j}
  \leq a\sum_{j\in F}\abs{u_j}
       +b\sum_{j\in F}\abs{u^\perp_j}
  \leq a\norm{u}_1+b\sqrt{k},
\end{equation}
where the last inequality is Cauchy-Schwarz. Consequently,
\begin{equation}\label{eq:subspace-threshold}
  \tau\leq\frac{a\norm{u}_1}{k}+\frac{b}{\sqrt{k}}
  \leq\frac{a}{\sqrt{k}}(\rho+x).
\end{equation}
Substituting this bound into \cref{eq:preparation-holder} yields
\begin{equation}\label{eq:subspace-discarded}
  \abs{\ip{u}{(I-\Pi_F)\theta}}
  \leq\norm{u}_1\tau
  \leq\frac{a\sqrt{\chi_\Omega}}{\sqrt{k}}(\rho+x)
  =a\rho(\rho+x).
\end{equation}
This proves \cref{eq:preparation-overlap-target}, and hence
\cref{eq:subspace-overlap}.

Now, we translate the obtained overlap bound into a bound on progress. Suppose $\rho x+\rho^2<1$. By \cref{eq:subspace-overlap}, $P_\Omega\omega\ne0$. Since the progress measure is scale invariant, $T_\Omega(\omega)=T_\Omega(\widehat\omega)$, where $\widehat\omega\coloneqq\omega/\norm{\omega}_2$. Truncation does not increase the Euclidean norm, so
\begin{equation}
  \norm{P_\Omega\widehat\omega}_2
  \geq\norm{P_\Omega\omega}_2
  \geq a(1-\rho x-\rho^2).
\end{equation}
It follows that
\begin{equation}
  T_\Omega(\omega)
  =\frac{\sqrt{1-\norm{P_\Omega\widehat\omega}_2^2}}
          {\norm{P_\Omega\widehat\omega}_2}\notag\leq\frac{\sqrt{1-a^2(1-\rho x-\rho^2)^2}}
             {a(1-\rho x-\rho^2)}=\frac{\sqrt{a^{-2}-1+2(\rho x+\rho^2)-(\rho x+\rho^2)^2}}
          {1-\rho x-\rho^2}.
\end{equation}
Recalling $x^2=b^2/a^2=(1-a^2)/a^2=a^{-2}-1$, we conclude that
\begin{equation}
  T_\Omega(\omega)
  \leq\frac{\sqrt{x^2+2\rho x+2\rho^2}}{1-\rho x-\rho^2}
  \leq\frac{x+\sqrt2\rho}{1-\rho x-\rho^2}.\qedhere
\end{equation}
\end{proof}

The scalar recurrence is controlled by the following lemma, proved in \cref{app:preparation-contraction}. 

\begin{restatable}[Contraction rate]{lemma}{contractionrate}\label{lem:preparation-recurrence}
Let $0<\gamma,\delta\leq1$ and $\rho>0$. Suppose that
$\rho\leq C\gamma\sqrt\delta$ for some $C\in(0,1/4]$. Define
\begin{equation}\label{eq:preparation-map}
  f(x)\coloneqq
  \frac{(1-\gamma)x+\sqrt2\rho}
       {1-\rho^2-\rho(1-\gamma)x},\qquad 0\leq x\leq1/\sqrt\delta.
\end{equation}
Then $f$ is well defined and nondecreasing on $[0,1/\sqrt\delta]$, and there is a unique fixed point $\xi_\star\in[0,1/\sqrt\delta]$ satisfying $0<\xi_\star<2\rho/\gamma$. If $\xi_0\in[0,1/\sqrt\delta]$ and $\xi_{t+1}\coloneqq f(\xi_t)$, then
\begin{equation}\label{eq:preparation-scalar-rate}
  \xi_t\in[0,1/\sqrt\delta],\qquad
  \abs{\xi_t-\xi_\star}\leq\frac{e^{-\gamma t/8}}{\sqrt\delta},
  \qquad t\geq0.
\end{equation}
\end{restatable}

\subsection{TPM correctness proof for ground state preparation}
\label{sec:proof-prep}

\begin{proof}[Proof of \cref{thm:preparation}]
For each iterate, write $x_t\coloneqq T_\Omega(v_t)$. Whenever $x_t\leq1/\sqrt{\delta_\Omega}$, the choice of $k$ gives
\begin{equation}
    \rho\leq\frac{\gamma\sqrt{\delta_\Omega}}4,\qquad
  \rho\leq\frac{\gamma\sqrt\eta}{5}.
\end{equation}
The first bound above implies
\begin{equation}\label{eq:preparation-denominator}
  \rho T_\Omega(Av_t)+\rho^2
  \leq\rho(1-\gamma)x_t+\rho^2
  \leq\frac\gamma4+\frac{\gamma^2}{16}
  \leq\frac{5\gamma}{16}<1.
\end{equation}
Thus \cref{lem:subspace-truncation} applies to
$Av_t/\norm{Av_t}_2$ and ensures that the next iterate has nonzero ground space projection. Scale invariance of $T_\Omega$ and the two one-step estimates yield
\begin{align}
  x_{t+1}
  =T_\Omega\bigl(\Trunc_k(Av_t)\bigr)
  \leq\frac{T_\Omega(Av_t)+\sqrt2\rho}
              {1-\rho T_\Omega(Av_t)-\rho^2}
  \leq\frac{(1-\gamma)x_t+\sqrt2\rho}
              {1-\rho(1-\gamma)x_t-\rho^2}
   =f(x_t),\label{eq:preparation-iterate-map}
\end{align}
where $f$ is as in \cref{eq:preparation-map}. The second inequality follows from \cref{lem:subspace-power}, since the fraction increases with $T_\Omega(Av_t)$ wherever its denominator is positive.

Recall that $\delta_\Omega=\norm{P_\Omega v_0}_2^2$. Thus, the initial progress satisfies
\begin{equation}\label{eq:preparation-smallness}
  x_0=\sqrt{\frac{1-\delta_\Omega}{\delta_\Omega}}
       \leq\frac1{\sqrt{\delta_\Omega}}.
\end{equation}
Next, we apply \cref{lem:preparation-recurrence} with
$\delta=\delta_\Omega$, $C=1/4$, and $\xi_0=x_0$. Since $f$ is
nondecreasing and the scalar iterates remain in
$[0,1/\sqrt{\delta_\Omega}]$, induction using
\cref{eq:preparation-denominator,eq:preparation-iterate-map} shows that all iterations are well defined and $0\leq x_t\leq\xi_t$. In particular,
\begin{equation}\label{eq:preparation-angle-bound}
  0\leq x_L\leq\frac{2\rho}{\gamma}
       +\frac{e^{-\gamma L/8}}{\sqrt{\delta_\Omega}}.
\end{equation}
Using $k\geq k_{\prep}$ and $L\geq L_{\prep}$, we obtain
\begin{equation}
  0\leq x_L\leq\frac25\sqrt\eta+\frac15\sqrt\eta\leq\sqrt\eta.
\end{equation}
Finally, write $v_L=a_Lu_L+b_Lu_L^\perp$, where $u_L\in\Omega$ and
$u_L^\perp\in\Omega^\perp$ are normalized. Since $x_L=b_L/a_L$, we have $b_L\leq x_L$ and
\begin{equation}
  \norm{P_\Omega v_L}_2^2=1-b_L^2\geq1-x_L^2\geq1-\eta.
\end{equation}
Moreover, $A\succeq\lambda P_\Omega$ implies
\begin{equation}
  \mu=\ip{v_L}{Av_L}
  \geq a_L^2\lambda=(1-b_L^2)\lambda
  \geq(1-x_L^2)\lambda\geq(1-\eta)\lambda.
\end{equation}
The upper bound $\mu\leq\lambda$ is immediate.
\end{proof}

We now give a comparison between \cref{thm:estimation,thm:preparation}. For any unit $v$ and lower bound $\Delta$ on $E_1-E_0$, we have
$1-\norm{P_\Omega v}_2^2\leq(\ip{v}{Hv}-E_0)/\Delta$. Therefore, an infidelity of at most $\eta$ can also be obtained by estimating the
energy to error at most $\Delta\eta$. Plugging $\varepsilon = \Delta\eta$ in \cref{thm:estimation} leads to a runtime that is worse than that obtained in \cref{thm:preparation} for the preparation task. On the other hand, the runtimes resulting from \cref{thm:estimation,thm:preparation} for the estimation task are incomparable. For more details, see  \cref{tab:complexity-comparison-intro,tab:complexity-comparison-intro} and \cref{eq:l1-norm-bound}.

%%%%%%%%%%%%%%%%%%%%%%%%%%%%%%%%%%%%%%%%%%%%%
\section{Comparing TPM to quantum algorithms}
\label{sec:dequantization}
%%%%%%%%%%%%%%%%%%%%%%%%%%%%%%%%%%%%%%%%%%%%%

In this section, we compare the runtimes of TPM (\cref{alg:tpm}) and Lin and Tong's quantum algorithm in \cite{preparation_lin_2020}. As discussed in the introduction, the algorithm in \cite{preparation_lin_2020} is an appropriate representative for all quantum algorithms for GLHP because it is near-optimal. To obtain the TPM runtime, we plug the settings of $k$ and $L$ derived in \cref{sec:results} into the bound \cref{eq:general-runtime} in \cref{sec:cost}. We assume that both algorithms start from an input Hamiltonian $H$ that is given as a linear combination of $M$ Pauli terms, $H = \sum_{i=1}^M \alpha_i P_i$.

For simplicity, we assume the ground state $\psi$ is nondegenerate. The runtimes are summarized in \cref{tab:complexity-comparison}. Compared to \cref{tab:complexity-comparison-intro} in the introduction, we have included the explicit dependence on the sparsity $S$ of the initial guess state. \cref{tab:complexity-comparison} shows that the runtime of TPM is polynomially related to that of the quantum algorithm when $\norm{\psi}_1 \leq \mathrm{poly}(n)$. Recall that $\delta$ is the initial overlap $\abs{\ip{\vin}{\psi}}^2$, and $\alpha$ is the Pauli coefficient sum $\sum_{i=1}^M \abs{\alpha_i}$, and $\Delta$ is a lower bound on the spectral gap.

\begin{table}[htbp]
  \centering
  \setlength{\tabcolsep}{14pt}
  \setlength{\aboverulesep}{8pt}
  \setlength{\belowrulesep}{8pt}
  \begin{tabular}{lcc}
    \toprule
    Ground state task & Quantum algorithm (\cite{preparation_lin_2020}) & TPM (this work) \\
    \midrule
    Energy estimation
      & $\dfrac{M\alpha}{\sqrt{\delta}\,\varepsilon} + \dfrac{S}{\sqrt{\delta}}$
      & $\norm{\psi}_1^2 \cdot \dfrac{M\alpha^3}{ \delta \varepsilon^3} +SM$ \\
    \addlinespace[8pt]
    Preparation
      & $\dfrac{M\alpha}{\sqrt{\delta}\,\Delta}  \log \Bigl(\dfrac 1 \eta \Bigr)  + \dfrac{S}{\sqrt{\delta}}$
      & $\norm{\psi}_1^2 \cdot \dfrac{M\alpha^3}{ \Delta^3} \max\Bigl(\dfrac{1}{\delta} , \dfrac{1}{\eta} \Bigr) +SM$ \\
    \bottomrule
  \end{tabular}
  \vspace{8pt}
  \caption{Runtimes of TPM and the quantum algorithm, including the dependence on the support size $S$ of $\vin$. Bounds hide constant and polylogarithmic factors.}
  \label{tab:complexity-comparison}
\end{table}

The runtimes in \cite{preparation_lin_2020} are given in terms of queries to two unitaries $U_H$ and $U_I$, where $U_H$ is a block encoding of $H/\alpha$ and $U_I$ prepares the initial guiding state $\ket{\vin}$. To convert query complexities into actual runtimes starting from the input $H$, we need to multiply them by the cost per query. To construct $U_H$, we use the linear combination of unitaries (LCU) method \cite{lcu_childs_2012,simulation_berry_2015}, which is treated pedagogically in \cite{algorithms_dalzell_2023}. Section 10.2 of~\cite{algorithms_dalzell_2023} shows that the cost per query to $U_H$ is then $\widetilde{O}(M)$. Using a similar approach, the cost per query to $U_I$ is $\widetilde{O}(S)$. For \emph{generic} $H$ and $\ket{\vin}$, there are nearly matching lower bounds of $\Omega(M)$ and $\Omega(S)$ by standard parameter counting arguments \cite{state_plesch_2011}.

\subsection{Ground state energy estimation}
\label{sec:dequantization-energy}

For the runtime of the quantum algorithm, we use Theorem 8 in \cite{preparation_lin_2020}. For the reader's convenience, we reproduce that theorem below using our notation. Recall that $H = \sum_{i=1}^M \alpha_i P_i$ is a linear combination of $M$ Pauli terms, and $\alpha \coloneqq \sum_{i=1}^M \abs{\alpha_i}$.

\begin{theorem}[{\cite[Theorem 8]{preparation_lin_2020}, paraphrased}]\label{thm:lt-gs}
Suppose we have a Hamiltonian $H$  given through its $(\alpha,\lceil\log_2 M\rceil,0)$-block-encoding $U_H$ using the LCU method~\cite{algorithms_dalzell_2023}. Also suppose we have an initial state $\ket{\phi_0}$ prepared by a circuit $U_I$, as well as the promise $\abs{\ip{\phi_0}{\psi_0}} \geq \sqrt {\delta}>0$. Then the ground state energy can be estimated to error at most $\varepsilon$ with probability at least $1-\vartheta$ using
    \begin{equation}
    \label{eq:energy-queries-UH}
        O\!\left(
            \frac{\alpha}{\sqrt{\delta}\,\varepsilon}
            \log\!\left(\frac{\alpha}{\varepsilon}\right)
            \log\!\left(\frac{1}{\sqrt{\delta}}\right)
            \log\!\left(
                \frac{\log(\alpha/\varepsilon)}{\vartheta}
            \right)
        \right)
    \end{equation}
    queries to $U_H$, and
    \begin{equation}
    \label{eq:energy-queries-UI}
        O\!\left(
            \frac{1}{\sqrt{\delta}}
            \log\!\left(\frac{\alpha}{\varepsilon}\right)
            \log\!\left(
                \frac{\log(\alpha/\varepsilon)}{\vartheta}
            \right)
        \right)
    \end{equation}
    queries to $U_I$, alongside some ancillary qubits and one- and two-qubit gates. The nonquery costs are subdominant and will be hidden in the total runtime bound below using $\tilde O$ notation.
\end{theorem}

\emph{Runtime of quantum algorithm.} We can bound the quantum algorithm's total runtime by summing up the total runtime of queries to $U_H$ and $U_I$. We only pay attention to the leading orders, which are the nonlogarithmic terms. The total runtime associated with $U_H$ is
\begin{equation}\label{eq:cost-UH-energy}
    \widetilde{O}\Bigl(\frac{\alpha}{\sqrt{\delta}\,\varepsilon}\Bigr)
    \cdot \widetilde{O}(M)
    = \widetilde{O}\Bigl(\frac{M\alpha}{\sqrt{\delta}\,\varepsilon}\Bigr),
\end{equation}
and that with $U_I$ is
\begin{equation}\label{eq:cost-UI-energy}
    \widetilde{O}\Bigl(\frac{1}{\sqrt{\delta}} \log \Bigl( \frac{\alpha}{\varepsilon}\Bigr)\Bigr)
    \cdot \widetilde{O}(S)
    = \widetilde{O}\Bigl(\frac{S}{\sqrt{\delta}}\log \Bigl( \frac{\alpha}{\varepsilon}\Bigr)\Bigr).
\end{equation}
Adding \cref{eq:cost-UH-energy} and \cref{eq:cost-UI-energy}, we obtain the overall runtime bound\footnote{The factor of $\log(\alpha/\epsilon)$ multiplying the second term can be suppressed since $\alpha/\epsilon$ appears in the first term.}
\begin{equation}\label{eq:lt-cost-energy}
    \widetilde{O}\Bigl(\frac{M\alpha}{\sqrt{\delta}\,\varepsilon} + \frac{S}{\sqrt{\delta}}\Bigr).
\end{equation}

\emph{Runtime of TPM.} To compute the TPM total cost, we simply substitute the expressions of $k_\est$ and $L_\est$ in \cref{eq:parameters} to the runtime in \cref{eq:general-runtime}. Including only the nonlogarithmic terms and ignoring constants, we obtain
\begin{equation}
    \widetilde O\bigl(MS+Mk L\bigr) = \widetilde{O} \Bigl(M \cdot S + M\frac{\norm{\psi}_1^2\alpha^2}{\delta \varepsilon^2} \cdot \frac{\alpha}{\varepsilon} \Bigr) = \widetilde{O} \Bigl(\frac{\norm{\psi}_1^2 M\alpha^3}{ \delta  \varepsilon^3} +SM\Bigr).
\end{equation}

\subsection{Ground state preparation}
\label{sec:dequantization-state}

For the runtime of the quantum algorithm, we use Corollary 9 in \cite{preparation_lin_2020}. For the reader's convenience, we reproduce that corollary below using our notation.

\begin{theorem}[{\cite[Corollary 9]{preparation_lin_2020}, paraphrased}]
\label{thm:lt-gs-preparation}
Suppose we have a Hamiltonian $H$ given through its $(\alpha,\lceil\log_2 M\rceil,0)$-block-encoding $U_H$ using the LCU method~\cite{algorithms_dalzell_2023}. Let $\lambda_0$ and $\lambda_1$ be its two smallest eigenvalues, and let $\ket{\psi_0}$ be a normalized ground state. Suppose we have a normalized initial state $\ket{\phi_0}$ prepared by a circuit $U_I$, together with known bounds $\abs{\ip{\phi_0}{\psi_0}}\geq\sqrt{\delta}>0$ and $\lambda_1-\lambda_0\geq\Delta>0$. Then an approximate ground state $\ket{\widetilde{\psi}_0}$ satisfying $\abs{\ip{\psi_0}{\widetilde{\psi}_0}}\geq1-\eta$ can be prepared with probability at least $1-\vartheta$ using\footnote{The condition $\abs{\ip{\psi_0}{\widetilde{\psi}_0}}\geq1-\eta$ implies $\abs{\ip{\psi_0}{\widetilde{\psi}_0}}^2\geq1-2\eta$, so by dividing $\eta$ by $2$, we are guaranteed that $\abs{\ip{\psi_0}{\widetilde{\psi}_0}}^2\geq1-\eta$, which is the standard one in this paper. This does not asymptotically change the query costs. We stated the theorem with  ``$\abs{\ip{\psi_0}{\widetilde{\psi}_0}}\geq1-\eta$'' simply because that is how it was stated in \cite{preparation_lin_2020}.}
    \begin{equation}
    \label{eq:preparation-queries-UH}
        O\!\left(
            \frac{\alpha}{\sqrt{\delta}\,\Delta}
            \left[
                \log\!\left(\frac{\alpha}{\Delta}\right)
                \log\!\left(\frac{1}{\sqrt{\delta}}\right)
                \log\!\left(
                    \frac{\log(\alpha/\Delta)}{\vartheta}
                \right)
                + \log\!\left(\frac{1}{\eta}\right)
            \right]
        \right)
    \end{equation}
    queries to $U_H$, and
    \begin{equation}
    \label{eq:preparation-queries-UI}
        O\!\left(
            \frac{1}{\sqrt{\delta}}
            \log\!\left(\frac{\alpha}{\Delta}\right)
            \log\!\left(
                \frac{\log(\alpha/\Delta)}{\vartheta}
            \right)
        \right)
    \end{equation}
    queries to $U_I$, alongside some ancillary qubits and one- and two-qubit gates. The nonquery costs are subdominant and will be hidden in the total runtime bound below using $\tilde O$ notation.
\end{theorem}

\emph{Runtime of quantum algorithm.} The total runtime associated with $U_H$ is
\begin{equation}\label{eq:cost-UH-state}
    \widetilde{O}\Bigl(\frac{\alpha}{\sqrt{\delta}\,\Delta} \Bigl(\log \Bigl(\frac \alpha \Delta \Bigr)+\log \Bigl(\frac 1 \eta \Bigr)\Bigr)\Bigr)
    \cdot \widetilde{O}(M)
    = \widetilde{O}\Bigl(\frac{M\alpha}{\sqrt{\delta}\,\Delta}  \Bigl(\log \Bigl(\frac \alpha \Delta \Bigr)+\log \Bigl(\frac 1 \eta \Bigr)\Bigr)\Bigr),
\end{equation}
and that with $U_I$ is
\begin{equation}\label{eq:cost-UI-state}
    \widetilde{O}\Bigl(\frac{1}{\sqrt{\delta}} \log \Bigl(\frac{\alpha}{\Delta} \Bigr)\Bigr)
    \cdot \widetilde{O}(S)
    = \widetilde{O}\Bigl(\frac{S}{\sqrt{\delta}}\log \Bigl(\frac{\alpha}{\Delta} \Bigr)\Bigr).
\end{equation}
Adding \cref{eq:cost-UH-state} and \cref{eq:cost-UI-state}, we obtain the overall runtime bound\footnote{The factor of $\log(\alpha/\Delta)$ multiplying the second term can be suppressed since $\alpha/\Delta$ appears in the first term.}
\begin{equation}\label{eq:lt-cost-prep}
    \widetilde{O}\Bigl(\frac{M\alpha}{\sqrt{\delta}\,\Delta}  \log \Bigl(\frac 1 \eta \Bigr) + \frac{S}{\sqrt{\delta}}\Bigr).
\end{equation}

\textbf{Runtime of TPM.} Now, we simply substitute the expressions of $k_{\prep}$ and $L_{\prep}$ in \cref{eq:preparation-parameters} to the runtime in \cref{eq:general-runtime}. Here, we also bound $\gamma$ in terms of $\alpha$ and $\Delta$ by $\gamma \geq \Delta/(2\alpha)$. Including only the nonlogarithmic terms  and ignoring constants, we obtain
\begin{equation}\label{eq:tpm-prep}
    \widetilde O\bigl(MS+Mk L\bigr) = \widetilde{O} \Bigl(M \cdot S + M\frac{\norm{\psi}_1^2 \alpha^2}{\Delta^2} \Bigl(\frac{1}{\delta} + \frac{1}{\eta} \Bigr) \cdot \frac{\alpha}{\Delta} \Bigr) = \widetilde{O} \Bigl(\frac{\norm{\psi}_1^2 M\alpha^3}{ \Delta^3} \Bigl(\frac{1}{\delta} + \frac{1}{\eta} \Bigr) +SM\Bigr).
\end{equation}

As we mentioned in the introduction, the quantum algorithm's runtime in \cref{eq:lt-cost-prep} should not be directly compared with that of TPM \cref{eq:tpm-prep}. This is because the quantum algorithm outputs a quantum state but TPM outputs a classical description of that state. In the next subsection, we bound the cost of doing tomography on the quantum state so that the quantum algorithm also outputs its classical description. This allows for a fair comparison. 

\subsection{State tomography with no explicit dependence on dimension.}
\label{sec:tomography}

In this subsection, we bound the cost of quantum state tomography by $\tilde{O}(\norm{\psi}_1^2/\eta^{3/2})$, which we also show is tight up to polylogarithmic factors in the query model. While our result is not the hardest to derive from \cite{tomography_vanapeldoorn_2022}, it may be of independent interest since there is no explicit dependence on the dimension of the state $\psi$. Instead, the dimension is replaced by the parameter $\norm{\psi}_1$. Existing tomography bounds, including those in \cite{tomography_vanapeldoorn_2022}, are usually stated with explicit dimension dependence.  The main result of this subsection is \cref{thm:tomography-cost}.

Before proceeding with the theorem statement and proof, we need to make some remarks on the connection between the precise output of \cite{preparation_lin_2020} and the assumptions in \cite{tomography_vanapeldoorn_2022} which we use in \cref{thm:tomography-cost}. The latter assume controlled access to the state-preparation unitary $U$ and its inverse $U^\dagger$, where $U$ prepares the state $\ket{v}$ on which we do tomography. However, the output of \cite{preparation_lin_2020} is not quite a unitary that prepares $\ket{v}$, an approximation of the ground state. Instead, it outputs a unitary $U_{\text{tot}}$ that prepares $\ket{v_{\mathrm{tot}}}$ which is a state that includes ancillas:
\begin{equation}
    U_{\text{tot}}\ket{0}_\text{anc}\ket{0}_{\text{syst}}=\sqrt{p_{\text{in}}}\ket{0}_{\text{anc}}\ket{v}_{\text{syst}}+\sqrt{1-p_{\text{in}}}\ket{v_{\perp}}=\ket{v_{\mathrm{tot}}}.
\end{equation}
Here, $\ket{\cdot}_{\text{anc}}$ is a state on ancillas and $\ket{\cdot}_{\text{syst}}$ is a state on the system such as the ground state approximation $\ket{v}$ (in the notation of \cite{preparation_lin_2020}, $U_{\text{tot}}=\mathrm{PROJ}(\mu,\Delta/(4\alpha),\gamma\epsilon)\cdot U_I$). $\ket{v_{\perp}}$ is a residual piece on the ancillas together with the system that is orthogonal to the first component $\ket{0}_{\text{anc}}\ket{v}_{\text{sys}}$. In particular, it does not contain any states with ancillas in the $\ket{0}_{\text{anc}}$ state so  projecting onto the $\ket{0}_{\text{anc}}$ ancilla state returns the desired state $\ket{v}_{\text{syst}}$. The parameter $p_{\text{in}}=\Omega(\delta)$ can be increased via amplitude amplification to make the $\ket{0}_{\mathrm{anc}}$ measurement more probable. Ref.~\cite{preparation_lin_2020} boosts the success weight to $\Omega(1)$ leading to the $1/\sqrt{\delta}$ query complexity to $U_{\text{tot}}$. The subtlety in applying the tomography technique to the output of \cite{preparation_lin_2020} is that the ancilla measurement precludes controlled access to a unitary that directly prepares $\ket{v}$.

To fill in this gap, we instead keep the ancillas for tomography and we defer the projection onto the $\ket{0}_{\text{anc}}$ state until the end. $U_{\text{tot}}$ is a unitary produced by \cite{preparation_lin_2020} that we indeed have controlled access to (to be precise, we will use the amplified version of $U_{\text{tot}}$, which we also have controlled access to). Instead of $\ket{\psi}$ as the unknown unit vector in Theorem~\ref{thm:tomography-cost}, we instead use $\ket{0}_{\text{anc}}\ket{\psi}_{\text{syst}}$, which has the same sparsity as $\ket{\psi}$; it is simply the original ground state tensored with ancillas in the state $\ket{0}_{\text{anc}}$. We need to amplitude amplify the $p_{\mathrm{in}}$ to $p\geq 1-\eta/4$, so that, combined with $|\langle\psi|v\rangle|^2\geq 1-\eta/4$, we get
\begin{equation}
    |(_\text{syst}\bra{\psi}_{\text{anc}}\bra{0})\ket{v_{\text{tot}}^p}|^2\geq 1-\eta/2
\end{equation}
as required by Theorem~\ref{thm:tomography-cost}. Here, $\ket{v_{\text{tot}}^p}$ is $\ket{v_{\text{tot}}} $ with $p_{\text{in}}$ amplified to $p$. Replacing the constant-success amplification in \cite{preparation_lin_2020} by a fixed-point amplitude amplification to success weight at least $1-\eta/4$ increases the number of queries to $U_{\text{tot}}$ and $U_{\text{tot}}^\dagger$ by only an $\mathcal{O}(\log (1/\eta))$ factor~\cite{search_yoder_2014}. Then the output of tomography gives a state $\ket{\tilde{v}_{\text{tot}}}$ satisfying
\begin{equation}
    |(_\text{syst}\bra{\psi}_{\text{anc}}\bra{0})\ket{\tilde{v}_{\text{tot}}}|^2\geq 1-\eta
\end{equation}
The final step is to apply the projection. Since we have a classical encoding of the state, this amounts to retaining the states with ancillas in the $\ket{0}_{\text{anc}}$ state, and renormalizing.

Because the tomography steps below are essentially unchanged besides relabeling $\ket{v}\to \ket{v_{\text{tot}}^p}$, $\ket{\tilde{v}}\to \ket{\tilde{v}_{\text{tot}}}$, etc., we will simplify notation by dropping the ancillas, with the understanding that to be precise we would have to keep the ancillas and project at the end after tomography.

\begin{theorem}\label{thm:tomography-cost}
Let $0<\eta < 1$. Let $\ket{v}\in\C^d$ be a quantum state approximating an unknown quantum state $\ket{\psi}\in\C^d$ such that $\abs{\ip{\psi}{v}}^2\geq1-\eta/2$. Let $U$ be a unitary that prepares $\ket{v}$ from a fixed initial state, i.e., $U\ket{0} = \ket{v}$. 

Then there is a quantum algorithm that outputs the classical description of a
unit vector $\tilde v$ satisfying
$\abs{\ip{\psi}{\tilde v}}^2\geq1-\eta$ with probability at least
$2/3$ using at most\footnote{The success probability of $2/3$ here is chosen for convenience. For any $0<\delta'<1$, the algorithm can also achieve success probability at least $1-\delta'$ with only an extra factor of $O(\log(1/\delta'))$ by standard boosting techniques.}
\begin{equation}
    \widetilde{O}(\norm{\psi}_1^2/\eta^{3/2})
\end{equation}
controlled queries to $U$ and its inverse.

Conversely, if $0<\eta \leq 1/64$, then for every $\chi\geq4$, there exists
a $\ket{\psi}$ with $\norm{\psi}_1=\Theta(\sqrt\chi)$ for which any quantum algorithm outputting $\tilde v$ as above requires at least 
\begin{equation}
    \Omega(\chi/\eta^{3/2})
\end{equation}
controlled queries to $U$ and its inverse.
\end{theorem}

Note that the theorem's upper bound is nontrivial even if $\ket{v} = \ket{\psi}$. The lower bound also holds even if $\ket{v} = \ket{\psi}$, as the proof will show.  We prove the upper bound first and then the lower bound.

\begin{proof}[Proof of upper bound]
The main tool for establishing the upper bound is \cite[Proposition 24 in arXiv v1]{tomography_vanapeldoorn_2022}, restated below for the reader's convenience.

\begin{theorem}[{\cite[Proposition 24 in arXiv v1]{tomography_vanapeldoorn_2022}, paraphrased}]\label{thm:tomography-helper}
Let $\ket{v}\in \C^d$ be a
quantum state. Let $U$ be a unitary that prepares $\ket{v}$ from a fixed initial state, i.e., $U\ket{0}=\ket{v}$. Let
$0<\varepsilon,\delta<1$, and let $R$ be a nonempty subset of the
indices such that
\begin{equation}
    \{j\in[d]:\abs{v_j}\geq\varepsilon\}
    \subseteq R\subseteq[d].
\end{equation}
Write $P\coloneqq\sum_{j\notin R}\abs{v_j}^2$ and $s\coloneqq\abs{R}$. There is a quantum algorithm that, with probability at least $1-\delta$, outputs the classical description of an $O\bigl((s+P/\varepsilon^2)\log(s)\log(1/\delta)\bigr)$-sparse vector $z\in\C^d$ satisfying $\norm{z-v}_\infty\leq\varepsilon$ using
\begin{equation}
    O\left(
      \left(
        \frac{\sqrt{s}}{\varepsilon}
        +\frac{\sqrt{P}}{\varepsilon^2}
        +\log\frac{\log(s+P/\varepsilon)}{\delta}
      \right)
      \log(s)
      \log\left(\frac{s+P/\varepsilon}{\delta}\right)
    \right)
\end{equation}
controlled queries to $U$ and its inverse.
\end{theorem}

Let $\tau \coloneqq \eta/(16\norm{\psi}_1)$. We use \cref{thm:tomography-helper} with
$\varepsilon=\tau$, since we will show that this guarantees the reconstructed state achieves squared fidelity at least $1-\eta$, with probability at least $2/3$ ($\delta = 1/3$). Under this choice, we will show that $s\in O(\norm{\psi}_1^2/\eta)$ and
$P\in O(\eta)$, so that a classical estimate $z$ satisfying $\norm{z-v}_\infty\leq\tau$ can be obtained using $\widetilde{O}(\norm{\psi}_1^2/\eta^{3/2})$ queries.

We first choose the global phase of $\psi$ so that
$\ip{\psi}{v}\geq0$ and define $e\coloneqq v-\psi$.
The fidelity assumption gives
\begin{equation}\label{eq:norm-e}
    \norm{e}_2^2
    =2-2\abs{\ip{\psi}{v}}
    \leq2-2\sqrt{1-\eta/2}
    \leq\frac{2\eta}{3}.
\end{equation}
Let $R=\{j \colon \abs{v_j}\geq\tau\}$ be the set of basis states of $v$ whose magnitudes are at least $\tau$. This set is nonempty, since otherwise $\abs{\ip{\psi}{v}}<\tau\norm{\psi}_1=\eta/16<\sqrt{1-\eta/2}$, contradicting the fidelity assumption. Every coordinate in $R$ satisfies either $\abs{\psi_j}\geq\tau/2$ or
$\abs{e_j}\geq\tau/2$, since $v=\psi+e$. Hence, in order to bound $s = |R|$, we define
$E=\{j\colon \abs{\psi_j}\geq\tau/2\}$ and $F=\{j\colon \abs{e_j}\geq\tau/2\}$. The cardinality of $E$ satisfies
\begin{equation}
    \abs{E}\frac{\tau}{2}
    \leq\sum_{j\in E}\abs{\psi_j}
    \leq\norm{\psi}_1
    \quad\Longrightarrow\quad
    \abs{E} \leq \frac{2\norm{\psi}_1}{\tau}
    =\frac{32\norm{\psi}_1^2}{\eta}.
\end{equation}
Similarly, using~\cref{eq:norm-e},
\begin{equation}
    \abs{F}\frac{\tau^2}{4}
    \leq\sum_{j\in F}\abs{e_j}^2
    \leq\norm{e}_2^2
    \leq\frac{2\eta}{3}
    \quad\Longrightarrow\quad
    \abs{F}\leq\frac{8\eta}{3\tau^2}
    =\frac{2048\norm{\psi}_1^2}{3\eta}.
\end{equation}
Since $R\subseteq E\cup F$, we obtain
\begin{equation}\label{eq:bigO-for-s}
    s=\abs{R}\leq\abs{E}+\abs{F}
    \leq\frac{32\norm{\psi}_1^2}{\eta}+\frac{2048\norm{\psi}_1^2}{3\eta}
    \leq O\Bigl(\frac{\norm{\psi}_1^2}{\eta}\Bigr).
\end{equation}

It remains to bound $P$. By the triangle inequality followed by the
Cauchy-Schwarz inequality,
\begin{equation}
    P
    =\sum_{j\notin R}\abs{v_j}^2\notag
    \leq\sum_{j\notin R}\abs{v_j}\abs{\psi_j}
      +\sum_{j\notin R}\abs{v_j}\abs{e_j}\notag
    \leq\tau\norm{\psi}_1
      +\Bigl(\sum_{j\notin R}\abs{v_j}^2\Bigr)^{1/2}
       \Bigl(\sum_{j\notin R}\abs{e_j}^2\Bigr)^{1/2}\notag
    \leq\tau\norm{\psi}_1+\sqrt P\,\norm{e}_2.
\end{equation}
Using $ab\leq(a^2+b^2)/2$ and~\cref{eq:norm-e} then gives
\begin{equation}\label{eq:bigO-for-P}
    P\leq\tau\norm{\psi}_1+\frac P2+\frac{\norm{e}_2^2}{2}
    \quad\Longrightarrow\quad
    P\leq2\tau\norm{\psi}_1+\norm{e}_2^2
    \leq\frac{\eta}{8}+\frac{2\eta}{3}
    =\frac{19\eta}{24} \leq O(\eta).
\end{equation}
Substituting~\cref{eq:bigO-for-s,eq:bigO-for-P} into
\cref{thm:tomography-helper} yields a sparse classical estimate $z$
satisfying $\norm{z-v}_\infty\leq\tau$ with probability at least
$2/3$, using a number of queries that is at most
\begin{equation}\label{eq:tomography-entrywise-cost}
    \widetilde{O}\Bigl(
      \frac{\sqrt{s}}{\tau}+\frac{\sqrt{P}}{\tau^2}
    \Bigr)
    =\widetilde{O}\Bigl(\frac{\norm{\psi}_1^2}{\eta^{3/2}}\Bigr).
\end{equation}

Now, let us translate the $\ell_\infty$ distance into the desired fidelity.
To obtain it, let $w$ be the Euclidean projection of $z$ onto
\begin{equation}
    C\coloneqq\{w\in\C^d:\norm{w}_1\leq\norm{\psi}_1\}.
\end{equation}
This means that $w$ is the vector in $C$ that is closest to $z$ in terms of the Euclidean distance. We will show that $\norm{w-\psi}_2^2 < \eta$ and that $\tilde v \coloneqq w/\norm{w}_2$ satisfies $\abs{\ip{\psi}{\tilde v}}^2 \geq 1-\eta$.

Notice that $C$ is closed, convex, and contains $\psi$. Hence, $w+t(\psi-w) \in C$ for all $t\in[0,1]$, and the distance to $z$ cannot decrease as $t$ increases. Therefore, we know that
\begin{equation}
    \frac{d}{dt} \left.\norm{(z-w)-t(\psi-w)}_2^2\right|_{t=0^+} =(-2\mathrm{Re}\ip{z-w}{\psi-w}+2t\norm{\psi-w}^2)_{t=0^+} \geq 0,
\end{equation}
which implies
\begin{equation}
    \operatorname{Re}\ip{z-w}{\psi-w}\leq0,
\end{equation}
and
\begin{equation}
    \norm{z-\psi}_2^2 = \norm{z-w}_2^2 + \norm{w-\psi}_2^2 -2\operatorname{Re}\ip{z-w}{\psi-w} \geq  \norm{z-w}_2^2 + \norm{w-\psi}_2^2.
\end{equation}

Writing $z-\psi = (v-\psi)+(z-v)$ and $z-w = (v-w) +(z-v)$, we obtain
\begin{align}
    \norm{w-\psi}_2^2
    &\leq \norm{v-\psi}_2^2-\norm{v-w}_2^2
      +2\operatorname{Re}\ip{z-v}{w-\psi}\notag\\
      &\leq\norm{v-\psi}_2^2
      +2\operatorname{Re}\ip{z-v}{w-\psi}\leq\norm{e}_2^2
      +2\norm{z-v}_\infty\norm{w-\psi}_1\notag\\
    &\leq\frac{2\eta}{3}
      +2\tau\bigl(\norm{w}_1+\norm{\psi}_1\bigr) \leq\frac{2\eta}{3}+4\tau\norm{\psi}_1
      =\frac{11\eta}{12}<\eta.
    \label{eq:tomography-projection}
\end{align}
Since $\psi$ is a unit vector and $\eta<1$, this implies $w\ne0$.
We can therefore normalize $w$ and return
$\tilde v \coloneqq w/\norm{w}_2$. The resulting fidelity satisfies
\begin{equation}\label{eq:tomography-fidelity}
    1-\abs{\ip{\psi}{\tilde v}}^2
    =\min_{c\in\C}\norm{\psi-cw}_2^2
    \leq\norm{\psi-w}_2^2<\eta.
\end{equation}
The projection and normalization are classical operations and require
no additional queries. Moreover, the projection cannot introduce a
nonzero coordinate outside the support of $z$, so the final estimate
can still be returned as a sparse vector. This proves the upper bound.
\end{proof}

\begin{proof}[Proof of lower bound.] Our approach for lower bounding the tomography cost is based on the proof of \cite[Lemma 51 in arXiv v1]{tomography_vanapeldoorn_2022}, slightly paraphrased below. In particular, we assume $x\in \{-1,1\}^d$, whereas \cite{tomography_vanapeldoorn_2022} assumes $x\in \{0,1\}^d$. However, this only changes  the lower bound by a factor of $2$ since an oracle encoding signs can be implemented by one encoding bits by multiplying $\epsilon$ by $2$.

\begin{lemma}[{from proof of \cite[Lemma 51 in  arXiv v1]{tomography_vanapeldoorn_2022}, paraphrased}]\label{lem:bit-recovery}
For $x\in\{-1,1\}^d$ and $\epsilon>0$, let the $\epsilon$-fractional phase oracle of $x$ be
\begin{equation}
    O_{\epsilon,x} \colon \ket{j} \mapsto e^{i\varepsilon x_j}\ket{j}, \quad \text{$j \in [d]$.}
\end{equation}
Then, recovering more than $(3/4)$-fraction of the signs of $x$ with probability at least $2/3$ requires $\Omega(d/\varepsilon)$ controlled queries to $O_{\epsilon,x}$ and its inverse in the worst case.
\end{lemma}

The proof proceeds by reduction from a sign-recovery problem. We consider a family of states $\ket{\psi_x}$, indexed by $x \in \{-1, 1\}^s$, all with the same $\ell_1$-norm. In each state, one basis state carries a fixed, dominant amplitude, and the remaining weight is spread evenly over $s$ other basis states whose relative signs are given by $x$. This structure means that a sufficiently accurate estimate of a state in the family determines all but a small fraction of the entries of $x$. By choosing the dominant amplitude appropriately, we show that any estimate that meets the target fidelity must recover more than $(3/4)$-fraction of the signs. This places us in the regime where \cref{lem:bit-recovery} applies. Finally, we show each $\ket{\psi_x}$ can be prepared using a constant number of queries to a fractional phase oracle encoding $x$, so the lower bound in \cref{lem:bit-recovery} becomes a lower bound on the cost of tomography.

Now, let us begin the actual proof. Define $a\coloneqq16\eta$ and $s\coloneqq\lfloor\chi/(64\eta)\rfloor$. Consider the following family of states specified by a string of signs $x\in\{-1,1\}^s$,
\begin{equation}\label{eq:tomography-hard-family}
    \ket{\psi_x}
    =\sqrt{1-a}\,\ket{0}
      +\sqrt{\frac as}\sum_{j=1}^s x_j\ket{j},
\end{equation}
in which all states in the family satisfy
\begin{equation}
    \norm{\psi_x}_1
    =\sqrt{1-a}+\sqrt{as}
    \leq1+\frac{\sqrt\chi}{2}
    \leq\sqrt\chi.
\end{equation}
Note that the last inequality uses $\chi\geq4$. Moreover, $s\geq\chi/(128\eta)$, so $\norm{\psi_x}_1\geq\sqrt{\chi/8}$ and hence $\norm{\psi_x}_1=\Theta(\sqrt\chi)$.

Suppose that the unknown state is $\ket{\psi_x}$ for some $x\in\{-1,1\}^s$. We will show that even if $U_x$ prepares $\ket{\psi_x}$ exactly (i.e., $U\ket{0} = \ket{\psi_x}$), obtaining the classical description of a unit vector $\tilde v$ as in the theorem requires $\Omega(\chi/\eta^{3/2})$ controlled queries to $U_x$ and its inverse. 

By assumption, we have $\abs{\ip{\tilde v}{\psi_x}}^2\geq 1-\eta$ with probability at least $2/3$. We now show that from $\tilde v$, we could recover at least $(3/4)$-fraction of the signs of $x$ with probability at least $2/3$. To that end, let $\psi_y$ be the state in the family that has the largest absolute overlap with $\tilde v$. Once we have $\tilde v$, we could obtain the bitstring $y$ without any further queries to $U_x$. We now show that at least $(3/4)$-fraction of the signs in $y$ must agree with $x$.

We first bound the overlap between $\psi_x$ and $\psi_y$ using trace distance, which we denote by $D$:
\begin{equation}
    D(u,w) \coloneqq \sqrt{1-\abs{\ip{u}{w}}^2}.
\end{equation}

The fidelity between $\tilde v$ and $\psi_x$ implies $D(\tilde v, \psi_x) \leq \sqrt \eta$. Since $\psi_y$ is the state in the family that is closest to $\tilde v$, we have $D(\tilde v, \psi_y) \leq D(\tilde v, \psi_x) \leq \sqrt \eta$. Applying the triangle inequality to the two distances, we obtain
\begin{equation} \label{eq:dist-ineq}
    D(\psi_x,\psi_y) \leq D(\tilde v, \psi_x) +D(\tilde v, \psi_y) \leq 2\sqrt \eta.
\end{equation}

Now, let us find the maximum number of sign mismatches between $x$ and $y$ permitted by \cref{eq:dist-ineq}. If $h$ signs differ between $x$ and $y$, then
$\ip{\psi_x}{\psi_y}=1-2ah/s$, so
\begin{equation}\label{eq:tomography-sign-errors}
    4\eta\geq D(\psi_x,\psi_y)^2
    =\frac{4ah}{s}\Bigl(1-\frac{ah}{s}\Bigr)
    \geq\frac{3ah}{s}.
\end{equation}
Here, the last inequality follows from $h\leq s$ and $a=16\eta\leq1/4$. Thus $h\leq s/12$, and successful tomography must recover at least $(11/12)$-fraction of the signs. Since $11/12$ is at least $3/4$, \cref{lem:bit-recovery} shows that, for any $\epsilon>0$, obtaining $\tilde v$ requires at least $\Omega(s/\epsilon)$ controlled queries to $O_{\epsilon,x}$ and its inverse. To convert this to a lower bound on controlled queries to $U$ and its inverse, it suffices to show that for $\epsilon = \Theta(\sqrt{a})$, the unitary $U_x$ can be implemented using one controlled query to $O_{\epsilon,x}$ and one controlled query to its inverse. We proceed to show how this can be done.

Write $\theta\coloneqq \arcsin \sqrt a $, which  is well defined since $\sqrt{a} = 4\sqrt{\eta} \leq 1$. Write $O_x$ for $O_{\theta,x}$ for convenience, recall that this means
\begin{equation}
    O_x \ket{j}=e^{i\theta x_j}\ket{j}.
\end{equation}

We construct $U_x$ satisfying
$U_x\ket{0}=\ket{\psi_x}$ as follows.
\begin{enumerate}
    \item Start with the state
    \begin{equation}
        \ket{u}\ket{+} = \frac{1}{\sqrt2}( \ket{u}\ket{0}+ \ket{u}\ket{1}), \quad \text{where} \quad \ket{u}=\frac1{\sqrt s}\sum_{j=1}^s\ket{j}.
    \end{equation}
    \item Apply $O_x$ when the ancilla is $0$
    and $O_x^\dagger$ when it is $1$. These two controlled queries
    produce
    \begin{equation}
        \frac1{\sqrt{2s}}\sum_{j=1}^s\ket{j}
        \bigl(e^{i\theta x_j}\ket{0}
              +e^{-i\theta x_j}\ket{1}\bigr).
    \end{equation}
        \item Apply a Hadamard gate followed by
    $\operatorname{diag}(1,-i)$ to the ancilla. Using $\cos(\theta x_j)=\cos\theta$ and
    $\sin(\theta x_j)=x_j\sin\theta$, this gives
    \begin{equation}\label{eq:tomography-weak-preparation}
        \cos\theta\,\ket{u}\ket{0}
        +\frac{\sin\theta}{\sqrt s}
          \sum_{j=1}^s x_j\ket{j}\ket{1}.
    \end{equation}
\end{enumerate}

The vectors $\ket{u}\ket{0}$ and
$\ket{j}\ket{1}$, $j=1,\ldots,s$, are orthonormal.
Hence there is a fixed unitary, independent of $x$, mapping them
to $\ket{0}$ and $\ket{j}$, respectively.
Since $\cos\theta=\sqrt{1-a}$ and $\sin\theta=\sqrt a$,
applying this unitary to
\eqref{eq:tomography-weak-preparation} produces
\begin{equation}
    \sqrt{1-a}\,\ket{0}
    +\sqrt{\frac as}\sum_{j=1}^s x_j\ket{j}
    =\ket{\psi_x}.
\end{equation}

Therefore, each query to $U_x$ can be implemented using one controlled query to $O_{\epsilon,x}$ and one controlled query to its inverse, as claimed. Therefore, \cref{lem:bit-recovery} yields a lower bound of $\Omega(s/\theta)$ on the number of controlled queries to $U_x$ and its inverse required for the tomography task. But $\theta = \arcsin \sqrt{a} = \arcsin(4\sqrt{\eta}) = \Theta(\sqrt{\eta})$ and $s = \Theta(\chi/\eta)$, so this translates into a lower bound of $\Omega(\chi/\eta^{3/2})$, as required.
\end{proof}

%%%%%%%%%%%%%%%%%%%%%%
\section{Discussion}
\label{sec:discussion}
%%%%%%%%%%%%%%%%%%%%%%

In this work, we showed that the GLHP can be dequantized by the TPM when the ground state is sparse, meaning it has polynomially bounded $\ell_1$-norm in the computational basis. There are several directions for future work. 

First, there may be room for polynomial improvements of TPM. One direction is to use Chebyshev polynomials rather than just powers of $A$. This method is known to optimally amplify the gap, and was used in \cite{preparation_lin_2020}. But it is unclear how to optimally implement this together with the truncation without incurring too much error. Another direction is to specialize to frustration-free (FF) Hamiltonians, where the Hamiltonian admits a local decomposition $H=\sum_iH_i$ such that the ground state simultaneously minimizes each term $H_i$. Note that a FF Hamiltonian may not be commuting, and may even be gapless. One might expect a speed-up at least of the ground state preparation protocol for FF Hamiltonians because local patches can be steered toward the true ground state in parallel. Indeed, \cite{renormalization_arad_2017} (mentioned again below) found better bounds in the FF case compared to general gapped Hamiltonians. The idea of steering toward the ground state in parallel on local patches is in essence the trick behind the ground state projector, a tool used for tightening various bounds for FF Hamiltonians~\cite{detectability_aharonov_2009,detectability_aharonov_2011,detectability_anshu_2016}. The difficulty in applying an ``in parallel'' steering of the state to the ground state on local patches is that even though $A$ is sparse, $\prod_{i\in \text{layer}}(\alpha_i-H_i)$ implementing a single layer of shifted local disjoint Hamiltonian terms (similar in structure to a layer of the approximate ground state projector) is not in general polynomially sparse. It would be interesting to try to find an efficient encoding of the state to handle this case, at least if we assume that the Hamiltonian is also translation invariant.

More generally, it would be useful to find efficient classical algorithms for GLHP that depend on more stable, local basis-independent measures of complexity. A generic local Hamiltonian has a ground state with  $\|\psi\|_1^2\sim\mathcal{O}(2^n)$, even if it is gapped and its ground state is area-law entangled. In fact, $\log\|\psi\|_1^2$ upper bounds the bipartite entanglement entropy. This is because the computational basis naturally gives a basis of states that is separable between any region $A$ and its complement $\bar{A}$, but the Schmidt decomposition used to calculate entanglement entropy optimizes this basis. Writing $\psi$ in the computational basis with coefficients $C_{x,y}$ gives
\begin{equation}\label{psidecomp}
    \ket{\psi} = \sum_{x,y}C_{x,y}\ket{x}_A\ket{y}_{\bar{A}}\to \log \|\psi\|_1^2=2\log\|\psi\|_1=2\log\sum_{x,y}|C_{xy}|.
\end{equation}

On the other hand, the Schmidt coefficients $s_i$ are the singular values of the coefficient matrix $C$, and the entanglement entropy is $S(A)=-\sum_is_i^2\log s_i^2$ since $s_i^2$ are the eigenvalues of the reduced density matrix $\rho_A$. Importantly, the R\'enyi entropies decrease with index, 
\begin{equation}
    S(A)=S_1(A)\leq S_{1/2}(A)=2\log\sum_is_i.
\end{equation}

Intuitively, a smaller R\'enyi index gives more weight to small eigenvalues, leading to a more spread-out distribution. Schmidt decomposition optimizes the decomposition between $A$ and $\bar{A}$ in the sense that
\begin{equation}\label{mindecomp}
    \sum_is_i=\min_{\ket{\psi} = \sum_jc_j\ket{a_j}_A\ket{b_j}_{\bar{A}}}\sum_j|c_j|
\end{equation}
for normalized states $\ket{a_j}_A$ supported on $A$ and $\ket{b_j}_{\bar{A}}$ supported on its complement. Note that $\ket{a_j}_A$ and $\ket{b_j}_{\bar{A}}$ do not need to be computational basis states or product states, allowing for much more flexibility in the minimization. Combining (\ref{psidecomp}) with (\ref{mindecomp}) and R\'enyi entropy monotonicity gives
\begin{equation}
    S(A)\leq \log \|\psi\|_1^2.
\end{equation}

Therefore, sparsity in the computational basis also implies small entanglement entropy. On the other hand, small entanglement entropy does not generally mean $\|\psi\|_1^2$ is small. For example, a product state in the $X$-basis has $\|\psi\|_1^2=\|\psi\|_0=2^n$ in the $Z$ computational basis.

Classical algorithms whose runtimes depend on entanglement entropy typically use ideas in the tensor-network framework such as DMRG \cite{renormalization_white_1992,renormalization_schollwock_2011}. But it is notoriously difficult to obtain rigorous performance guarantees on DMRG, especially beyond one spatial dimension, even when we know that the ground state should have an efficient tensor network representation. Refs.~\cite{polynomial_landau_2015,renormalization_arad_2017} obtained rigorous results in one dimension using a different approximate ground-state projector approach, but these guarantees rely on one-dimensional structure and do not directly extend to general higher-dimensional systems. The benefit of the analysis in this work is that we can rigorously upper bound, independent of spatial dimension, the resources necessary to get within some error in ground state energy estimation and preparation. It would be interesting to obtain analogous rigorous bounds for algorithms based on entanglement entropy or other more local basis-independent measures of complexity, to dequantize the GLHP and similar problems in those settings.

%%%%%%%%%%%%%%%%%%%%%%%%%%
\section*{Acknowledgments}
%%%%%%%%%%%%%%%%%%%%%%%%%%

We thank Sergey Bravyi and David Gosset for informing us of the TPM algorithm in \cite{accuracy_kirby_2026}. We thank Mona Berciu and Chen Greif for insightful discussions. S.L.C. was supported by the Work Learn International Undergraduate Research Award (WLIURA) through UBC Computer Science, the Quantum Pathways program of the UBC Blusson Quantum Matter Institute, and the QuantumBC CREATE Undergraduate Summer Research Top Up. D.W. is supported by NSERC Grants CRC-2023-00039, RGPIN-2024-06493, and DGECR-2024-00113. C.Z. is supported by NSERC grants RGPIN-2026-07143 and DGECR-2026-00535.

\textbf{AI disclosure.} The authors conjectured that the convergence of TPM could be given in terms of the $\ell_1$-norm of the true ground state instead of the $\ell_0$-norm. We also conjectured that the complexity of the estimation problem could be made independent of the spectral gap. Given \cite{accuracy_kirby_2026} as reference, GPT-5.6 Sol (and above) proved our conjectures, which became the basis of our main theorems (\cref{thm:estimation,thm:preparation}). We verified these proofs and then independently wrote them up. In the process, we improved their presentations, in particular, by making them more modular and similar to each other. The authors also formulated the state tomography task in \cref{thm:tomography-cost} and conjectured that it could be performed using $\poly(\norm{\psi}_1)/\poly(\eta)$ queries to the state preparation oracle. Given \cite{tomography_vanapeldoorn_2022} as reference, GPT-6 Astra gave an initial proof of this conjecture, which we digested and rewrote in human-readable form.

\bibliographystyle{unsrt} % Number references in order of citation
\bibliography{references} % Loads references.bib

\appendix
\crefalias{section}{appendix}

\newpage

%%%%%%%%%%%%%%%%%%%%%%%%%%%%%%%%%%%%%%%
\section{Proof of the contraction rate}
\label{app:preparation-contraction}
%%%%%%%%%%%%%%%%%%%%%%%%%%%%%%%%%%%%%%%

For convenience, we reproduce the contraction rate lemma below.
\contractionrate*

\begin{proof}[Proof of \cref{lem:preparation-recurrence}]
For $x\in[0,1/\sqrt\delta]$,
the assumption $\rho\leq C\gamma\sqrt\delta$ gives
\begin{equation}\label{eq:contraction-domain}
  \rho^2+\rho(1-\gamma)x
  \leq C^2\gamma^2\delta+C\gamma(1-\gamma)
  \leq(C^2+C)\gamma
  \leq\frac{5\gamma}{16}<1.
\end{equation}
Hence $f$ is well defined throughout $[0,1/\sqrt\delta]$. We now prove the lemma in three steps. 

\emph{Existence of a fixed point in $(0,2\rho/\gamma)$.}
We have $f(0)=\sqrt2\rho/(1-\rho^2)>0$. Set $r\coloneqq2\rho/\gamma$.
Since $r\leq2C\sqrt\delta\leq\sqrt\delta/2\leq1/\sqrt\delta$,
the denominator is positive at $r$, and hence
\begin{equation}
  f(r)<r
  \quad\Longleftrightarrow\quad
  2\rho^2\Bigl(1+\frac{2(1-\gamma)}{\gamma}\Bigr)
  <(2-\sqrt2)\gamma.
\end{equation}
The left-hand side is at most $6C^2\gamma$, with
$6C^2\leq3/8<2-\sqrt2$. Thus $f(r)<r$, and continuity gives a
fixed point $\xi_\star\in(0,r)$. Moreover,
\begin{equation}\label{eq:preparation-derivative}
  f'(x)=
  \frac{(1-\gamma)+\rho^2(1-\gamma)(\sqrt2-1)}
       {\bigl(1-\rho^2-\rho(1-\gamma)x\bigr)^2}
  \geq0.
\end{equation}
Thus $f$ is nondecreasing; the third step will prove uniqueness.

\emph{Convergence of $\xi_t$ to $\xi_\star$.}
The denominator is positive at $1/\sqrt\delta$, so
\begin{equation}
  f(1/\sqrt\delta)<1/\sqrt\delta
  \quad\Longleftrightarrow\quad
  \rho\bigl(\sqrt2\delta+1-\gamma+\rho\sqrt\delta\bigr)
  <\gamma\sqrt\delta.
\end{equation}
The parenthesized expression is less than $4$, so the left-hand side is less than $4C\gamma\sqrt\delta\leq\gamma\sqrt\delta$. Thus $f(1/\sqrt\delta)<1/\sqrt\delta$; together with $f(0)>0$ and monotonicity, this gives $f([0,1/\sqrt\delta])\subseteq[0,1/\sqrt\delta]$. Every $\xi_t$ therefore lies in $[0,1/\sqrt\delta]$. The sequence is nondecreasing if $\xi_1\geq\xi_0$ and nonincreasing otherwise. It is bounded, so it converges. By continuity, its limit is a fixed point; the third step identifies this limit as $\xi_\star$.

\emph{Bounding the rate of convergence.}
By \cref{eq:contraction-domain,eq:preparation-derivative}, for
$x\in[0,1/\sqrt\delta]$ we have
\begin{equation}\label{eq:contraction-derivative-bound}
  0\leq f'(x)
  \leq\frac{(1-\gamma)(1+\gamma^2/16)}{(1-5\gamma/16)^2}
  \leq1-\frac{3\gamma}{8}
  \leq e^{-\gamma/8}.
\end{equation}
The penultimate inequality follows from the identity
\begin{equation}
  (1-3\gamma/8)(1-5\gamma/16)^2
  -(1-\gamma)(1+\gamma^2/16)
  =\frac{\gamma^2(552+53\gamma)}{2048}\geq0.
\end{equation}
The mean value theorem now gives
\begin{equation}
  \abs{f(x)-f(y)}\leq e^{-\gamma/8}\abs{x-y},
  \qquad x,y\in[0,1/\sqrt\delta].
\end{equation}
Since $e^{-\gamma/8}<1$, the fixed point in $[0,1/\sqrt\delta]$
is unique. Applying this estimate recursively yields
\begin{equation}
  \abs{\xi_t-\xi_\star}
  \leq e^{-\gamma t/8}\abs{\xi_0-\xi_\star}
  \leq\frac{e^{-\gamma t/8}}{\sqrt\delta}.\qedhere
\end{equation}
\end{proof}

\end{document}